\documentclass[10pt,journal,compsoc]{IEEEtran}

\usepackage{times,amsmath,epsfig}
\usepackage[linesnumbered,ruled,vlined]{algorithm2e}
\usepackage{booktabs,graphicx}

\usepackage{subfigure}
\usepackage{xspace}
\usepackage{multirow}
\usepackage{graphicx}
\usepackage{amsthm}

\usepackage{tabularx}

\usepackage{amssymb}

\usepackage{color}
\usepackage{colortbl}
\usepackage[usenames,dvipsnames]{xcolor}
\usepackage{xcolor}

\usepackage{flushend}
\usepackage{algorithmic}
\usepackage{microtype}

\usepackage{hyperref}

\newtheorem{myDef}{Definition}
\newtheorem{myTheo}{Theorem}
\newtheorem{myLem}{Lemma}

\newtheorem{myExmp}{Example}

\newcommand{\TT}[1]{\textcolor{black}{#1}}
\newcommand{\TTao}[1]{\textcolor{black}{#1}}

\ifCLASSOPTIONcompsoc
  \usepackage[nocompress]{cite}
\else
  \usepackage{cite}
\fi

\begin{document}

\title{Incremental Graph Computation: Anchored Vertex Tracking in Dynamic Social Networks} 

\author{Taotao~Cai,
        Shuiqiao~Yang,
        Jianxin~Li$^*$, 
		Quan Z. Sheng,
		Jian~Yang, \\
		Xin~Wang,
		Wei Emma~Zhang,
		and Longxiang~Gao

\IEEEcompsocitemizethanks{\IEEEcompsocthanksitem Jianxin Li is with Deakin University, Melbourne, Australia. Jianxin Li is the corresponding author.
E-mail: jianxin.li@deakin.edu.au 
\IEEEcompsocthanksitem Taotao Cai, Quan Z. Sheng, and Jian Yang are with Macquarie University, Sydney, Australia.
E-mail: \{taotao.cai, michael.sheng, jian.yang\}@mq.edu.au
\IEEEcompsocthanksitem Shuiqiao Yang is with University of New South Wales, Sydney, Australia. Email: shuiqiao.yang@unsw.edu.au.
\IEEEcompsocthanksitem Xin Wang is with College of Intelligence and Computing, Tianjin University, Tianjin, China. E-mail: wangx@tju.edu.cn
\IEEEcompsocthanksitem Wei Emma Zhang is with the University of Adelaide, Adelaide, Australia.
E-mail: wei.e.zhang@adelaide.edu.au
\IEEEcompsocthanksitem Longxiang Gao is with Qilu University of Technology (Shandong Academy of Sciences) and Shandong Computer Science Center (National Supercomputer Center in Jinan). E-mail: gaolx@sdas.org.
\IEEEcompsocthanksitem  Taotao Cai and Shuiqiao Yang are the joint first authors.}

}




\IEEEtitleabstractindextext{%

\begin{abstract}
User engagement has recently received significant attention in understanding the decay and expansion of communities in many online social networking platforms. When a user chooses to leave a social networking platform, it may cause a cascading dropping out among \TT{her} friends. In many scenarios, it would be a good idea to persuade critical users to stay active in the network and  prevent such a cascade because critical users can have significant influence on user engagement of \TT{the} whole network. Many user engagement studies \TT{have been conducted} to find a set of critical \textit{(anchored)} users in the static social network. However, \TT{social networks are highly dynamic and their structures are continuously evolving}. In order to fully utilize the power of anchored users in evolving networks, existing studies have to mine multiple sets of anchored users at different times, which incurs an expensive computational cost. \TT{To better understand user engagement in evolving network, we target a new research problem called \textit{Anchored Vertex Tracking} (AVT) in this paper, aiming} to  
track the anchored users at each timestamp of evolving networks. Nonetheless, it is nontrivial to handle the AVT problem which we have proved to be NP-hard. To address the challenge, we develop a greedy algorithm inspired by the previous anchored $k$-core study in the static networks. Furthermore, we design an incremental algorithm to efficiently solve the AVT problem by utilizing the smoothness of the network structure's evolution. The extensive experiments conducted on real and synthetic datasets demonstrate the performance of our proposed algorithms and the effectiveness in solving the AVT problem.    
\end{abstract}

}

\maketitle

\IEEEdisplaynontitleabstractindextext

\IEEEpeerreviewmaketitle

\section{Introduction}
\label{sec:intro}

\IEEEPARstart{I}{N} recent years, user engagement has become a hot research topic in network science, arising from a plethora of online social networking and social media applications, such as \textit{Web of Science Core Collection}, \textit{Facebook}, and \textit{Instagram}. 
Newman~\cite{newman2001clustering} studied the collaboration of users in a collaboration network, and found that the probability of collaboration between two users \TT{is} highly related to the number of common neighbors of the selected users. Kossinets and Watts~\cite{doi:10.1086/599247,Kossinets88} verified that two users who have numerous common friends are more likely to be friends by investigating a series of social networks. 
Cannistraci et al.~\cite{Cannistraci2013} presented that two social network users are more likely to become friends if their common neighbors are members of a local community, and the strength of their relationship relies on the number of their common neighbors in the community.
\TTao{Centola et al.~\cite{centola2010spread} stated that in the presence of high clustering (i.e., $k$-core), any additional adoption of messages is likely to produce more multiple exposures than in the case of low clustering. Each additional exposure significantly increases the chance of message adoption.}
\TTao{Weng et al.~\cite{weng2013virality} pointed out that people are more susceptible to the information from peers in the same community. 
This is because the people in the same community sharing similar characteristics naturally establish more edges among them.}
Moreover, Laishram et al.~\cite{DBLP:conf/sdm/LaishramSEPS20} mentioned that the incentives for keeping users' engagement on a social network platform partially \TT{depends} on how many friends they can keep in touch with. Once the users' incentives are low, they may leave the platform. The decreased engagement of one user may affect others' engagement incentives, further causing them to leave. 
Considering a model of user engagement in a social network platform, where the participation of each user is motivated by the number of engaged neighbors. The user engagement model is a natural equilibrium corresponding to the $k$-core of the social network, where $k$-core is a popular model to \TT{identify} the maximal subgraph in which every vertex has at least $k$ neighbors.
The leaving of some critical users may cause a cascading departure from the social network platform. Therefore, the efforts of user engagement studies~\cite{DBLP:conf/icalp/BhawalkarKLRS12,DBLP:conf/asunam/RossettiPKPGD15, DBLP:conf/cikm/MalliarosV13, DBLP:journals/siamdm/BhawalkarKLRS15, DBLP:journals/pvldb/ZhangZZQL17} have been devoted to finding the crucial (anchored) users who significantly impact the formation of social communities and the operations of social networking platforms. \TT{In particular, Bhawalkar et al.~\cite{DBLP:conf/icalp/BhawalkarKLRS12} first studied the problem of anchored $k$-core, aiming to retain (anchor) some users with incentives to ensure they will not leave the community modeled by $k$-core, such that the maximum number of users will further remain engaged in the community.}



\TT{The previous studies of anchored $k$-core~\cite{DBLP:conf/icalp/BhawalkarKLRS12,DBLP:journals/pvldb/ZhangZZQL17,DBLP:conf/sdm/LaishramSEPS20} for user engagement have benefited many real-life applications, such as revealing the evolution of the community's decay and expansion in social networks.}
However, most of the previous anchored $k$-core researches dedicated to user engagement depend on a strong assumption - social networks are modelled as static graphs. This simple premise rarely reflects the evolving nature of social networks, of which the topology often evolves over time in real world~\cite{DBLP:conf/sdm/ChenSHX15, DBLP:conf/kdd/LeskovecBKT08}. Therefore, for a given dynamic social network, the anchored users selected at an earlier time may not be appropriate to be used for user engagement in the following time due to the evolution of the network.

To better understand user engagement in evolving networks, one possible way is to re-calculate the anchored users after the network structure is dynamically changed. A natural question is how to select $l$ anchored users at each timestamp of an evolving social network, 
\TT{so that the community size will be maximum when we persuade these $l$ users to keep engaged in the community of each timestamps.}
We refer this problem as \textit{Anchored Vertex Tracking} (AVT), which aims to find a series of anchored vertex sets with each set size limited to $l$. 
\TT{In other words, under the above problem scenario, it requires performing the anchored $k$-core query at each timestamp of evolving networks.}
By solving the proposed AVT problem, we can efficiently track the anchored users to improve the effectiveness of user engagement in evolving networks. 

Tracking the anchored vertices could be very useful for many practical applications, such as sustainable analysis of social networks, impact analysis of advertising placement, and social recommendation. 
Taking the impact analysis of advertising placement as an example. Given a social network, the users' connection often evolves, which leads to the dynamic change of user influences and roles. The AVT study can continuously track the critical users to locate a set of users who favor propagating the advertisements at different times. In contrast, traditional user engagement methods like OLAK~\cite{DBLP:journals/pvldb/ZhangZZQL17} and RCM~\cite{DBLP:conf/sdm/LaishramSEPS20} only work well in static networks. Therefore, AVT can deliver timely support of services in many applications.  
Here, we utilize an example in Figure~\ref{fig:Motivation} to explain the AVT problem in details. 


\begin{figure}[t]
	\centering
		\includegraphics[scale=0.50]{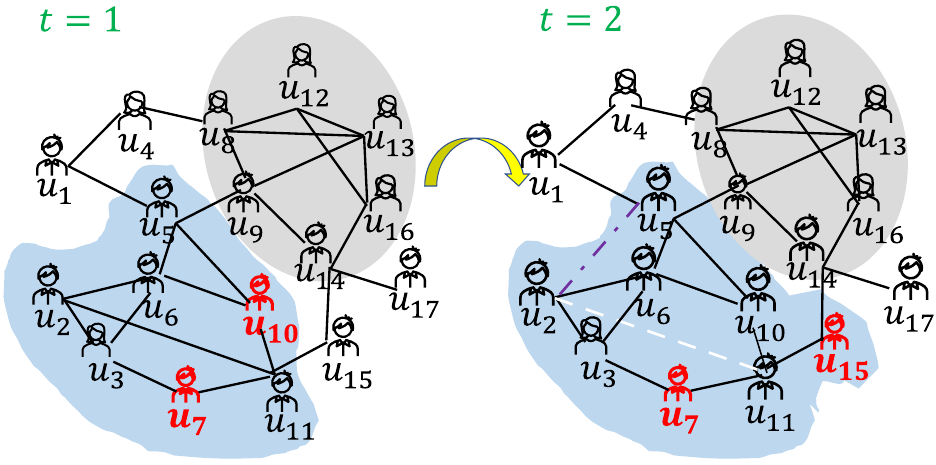}
		\caption{\TT{An example of Anchored Vertex Tracking (AVT).}}
		\vspace{-4mm}
	\label{fig:Motivation}
	\vspace{-3mm}
\end{figure}

\begin{myExmp}
\TT{Figure~\ref{fig:Motivation} presents a reading hobby community with 17 users and their friend relationships over two continuous periods. The number of a user's friends in the network reflects his willingness to engage. If one user has many friends (neighbors), the user would be willing to remain engaged in the community. Moreover, if a user leaves the community, it will weaken their friends' willingness to remain engaged in the community. 
According to the above engagement model with number of friends $k = 3$ (e.g., a user keep engaged in the group iff at least $3$ of his/her friends remaining engaged in the same community), $3$-core of the network at timestamp $t=1$ would be $\{u_8, u_9, u_{12},u_{14},u_{16}\}$ (covered by gray color). If we motivate users $\{u_7, u_{10}\}$ (e.g., red icons with friends less than 3) to keep engaged in the network at the timestamp $t=1$, then the users $\{u_2,u_3,u_5,u_6,u_{11}\}$ will remain engaged in the community because they have three friends in the reading hobby community now. Therefore, the number of $3$-core users would increase from 5 (gray) to 12 (gray \& blue). With the evolution of the network, at the timestamp $t=2$, a new relationship between users $u_2$ and $u_5$ is established (purple dotted line) while the relationship of users $u_2$ and $u_{11}$ is broken (white dotted line). Under this situation, the number of 3-core users will increase from $5$ to $14$ if we persuade users $\{u_7, u_{15}\}$ to keep the engagement in the community; However, the $3$-core users would only increase to $11$ once we motivate users $\{u_7,u_{10}\}$ to keep engaged. Therefore, the optimal users (called ``anchor") we selected to 
keep engaging may vary in different timestamps while the network evolves. 
}
\end{myExmp}

\noindent
\textbf{Challenges.} Considering the dynamic change of social networks and the scale of network data, it is infeasible to directly use the existing methods ~\cite{DBLP:journals/siamdm/BhawalkarKLRS15, DBLP:conf/aaai/ChitnisFG13, DBLP:journals/pvldb/ZhangZZQL17,DBLP:conf/sdm/LaishramSEPS20} \TT{of the anchored $k$-core problem} to compute the anchored user set for every timestamp. 
We prove that the AVT problem is NP-hard. To the best of our knowledge, there is no existing work to solve the AVT problem, particularly when the number of timestamps is large. 

To conquer the above challenges, we first develop a Greedy algorithm by extending the previous anchored $k$-core study in the static graph~\cite{DBLP:conf/icalp/BhawalkarKLRS12,DBLP:journals/pvldb/ZhangZZQL17}. However, the Greedy algorithm is expensive for large-scale social network data. Therefore, we optimize the Greedy algorithm in two aspects: (1) reducing the number of potential anchored vertices; and (2) accelerating computation of followers.
To further improve the efficiency, we also design an incremental algorithm by utilizing the smoothness of the network structure's evolution. 

\vspace{2mm}
\noindent
\textbf{Contributions.} We state our major contributions as follows:
\begin{itemize}
\item We formally define the problem of AVT and explain the motivation of solving the problem with real applications. 
\item We propose a Greedy algorithm by extending the core maintenance method in~\cite{DBLP:conf/icde/ZhangYZQ17} to tackle the AVT problem. Besides, we build several pruning strategies to accelerate the Greedy algorithm. 
\item We develop an efficient incremental algorithm by utilizing the smoothness of the network structure's evolution and the well-designed fast-updating core maintenance methods in evolving networks. 
\item We conduct extensive experiments to demonstrate the  efficiency and effectiveness of proposed approaches using real and synthetic datasets. 
\end{itemize}  

\noindent
\textbf{Organization.} 
We present the preliminaries in Section~\ref{sec:problem}. Section~\ref{sec:analysis} formally defines the AVT problem. We propose the Greedy algorithm in Section~\ref{sec:alg}, and further develop an incremental algorithm to solve the AVT problem more efficiently in Section~\ref{sec:improved}. The experimental results are reported in Section~\ref{sec:exp}. Finally, we review the related works in Section~\ref{sec:related}, and conclude the paper in Section~\ref{sec:Con}.

\section{Preliminaries} \label{sec:problem}

We define an \TT{undirected} evolving network as a sequence of graph snapshots  $\mathcal{G} = \{G_{t}\}_1^T$, and $\{1,2,.., T\}$ is a finite set of time points. We assume that the network snapshots in $\mathcal{G}$ share the same vertex set. Let $G_t$ represent the network snapshot at timestamp $t\in [1,T]$, where $V$ and $E_t$ are the vertex set and edge set of $G_t$, respectively. \TT{Similar to~\cite{DBLP:journals/tkde/JiaLDZGXZ21,das2019incremental}, we can create ``dummy" vertices at each time step $t$ to represent the case of vertices joining or leaving the network at time $t$ (e.g., $V = \cup_{t+1}^{T} V^t$ where $V^t$ is the set of vertices truly exist at $t$).} Besides, we set $nbr(u,G_t)$ as the set of vertices adjacent to vertex $u\in V$ in $G_t$, and the degree $d(u,G_t)$ represents the number of neighbors for $u$ in $G_t$, $i.e., |nbr(u,G_t)|$. 
Table~\ref{tab:symbol} summarizes the mathematical notations frequently used throughout this paper.

\subsection{Anchored $k$-core}
We first introduce \TT{the} notion of \textbf{$k$-core}, which has been widely used to describe the cohesiveness of subgraph. 
\begin{myDef}[\TT{k-core~\cite{DBLP:journals/corr/cs-DS-0310049}}]\label{def:k-core}
Given \TT{an undirected} graph $G_t$, the $k$-core of $G_t$ is the maximal subgraph in $G_t$, denoted by $C_k$, in which the degree of each vertex in $C_k$ is at least $k$. 
\end{myDef}

The $k$-core of a graph $G_t$, can be computed by repeatedly deleting all vertices (and their adjacent edges) with the degree less than $k$. The process of the above $k$-core computation is called \textbf{core decomposition}~\cite{DBLP:journals/corr/cs-DS-0310049}, which is described in Algorithm~\ref{alg:k-core}.

For a vertex $u$ in graph $G_t$, the \textbf{core number} of $u$, denoted as $core(u)$, is the maximum value of $k$ such that $u$ is contained in the $k$-core of $G_t$. Formally,

\begin{table}[t!]
  \begin{center}
    \caption{Notations Frequently Used in This Paper} 
    \label{tab:symbol}
    \vspace{-4mm}
    \scalebox{0.9}{
    \begin{tabular}{|p{1.5cm}<{}|p{5.5cm}<{}|}
    \hline
      \textbf{Notation}         & \textbf{Definition} \\
      \hline
       $\mathcal{G}$               & an \TT{undirected} evolving graph \\\hline
       $G_t$             & the snapshot graph of $\mathcal{G}$ at time instant $t$ \\\hline
       $V$; $E_t$                 & the vertex set and edge set of $G_t$ \\\hline 
       $nbr(u,G_t)$               & the set of adjacent vertices of $u$ in $G_t$ \\\hline
       $d(u,G_t)$                 & the degree of $u$ in $G_t$ \\\hline
       $deg^+(u)$              & the remaining degree of $u$ \\\hline
       $deg^-(u)$               & the candidate degree of $u$ \\\hline
       $C_k$                 & the $k$-core subgraph \\\hline
       $O(G_t)$                   & the $K$-order of $G_t$ where $O(G_t) = \{\mathcal{O}_1, \mathcal{O}_2,...\}$ \\\hline
       $C_k(\mathcal{S}_t)$  & the anchored $k$-core that anchored by $\mathcal{S}_t$ \\\hline
       $\mathcal{S}_t$          & the anchored vertex set of $G_t$ \\\hline
       $F_k(u,G_t)$              & followers of an anchored vertex $u$ in $G_t$ \\\hline
       $F_k(\mathcal{S}_t,G_t)$ & followers of an anchored vertex set $\mathcal{S}_t$ in $G_t$ \\\hline
       $E^+$; $E^-$               & the edges insertion and edges deletion from graph snapshots $G_{t-1}$ to $G_t$ \\\hline
       $mcd(u)$                  & the max core degree of $u$ \\\hline
    \end{tabular}
    }
  \end{center}
  \vspace{-6mm}
\end{table}

\begin{myDef}[Core Number]\label{def:corenumber}
Given \TT{an undirected} graph $G_t=(V,E_t)$, for a vertex $u\in V$, its core number, denoted as $core(u)$, is defined as $core(u,G_t)=max\{k: u\in C_k\}$.   
\end{myDef}

When the context is clear, we use $core(u)$ instead of $core(u,G_t)$ for the sake of concise presentation.
\begin{myExmp}
Consider the graph snapshot $G_1$ in Figure~\ref{fig:Motivation}. The subgraph $C_3$ induced by vertices $ \{u_8, u_9, u_{12}, u_{13}, u_{16}\}$ is the $3$-core of $G_1$. This is because every vertex in the induced subgraph has a degree at least $3$. Besides, there does not exist a $4$-core in $G_1$. Therefore, we have $core(v) = 3$ for each vertex $v\in C_3$. 
\end{myExmp}

If a vertex $u$ is \textbf{anchored}, in this work, it supposes that such vertex meets the requirement of $k$-core regardless of the degree constraint. The anchored vertex $u$ may lead to add more vertices into $C_k$ due to the contagious nature of \textit{$k$-core} computation. These vertices are called as \textbf{followers} of $u$.

\begin{myDef}[Followers]\label{def:followers}
Given \TT{an undirected} graph $G_t$ and an anchored vertex set $S_t$, the followers of $S_t$ in $G_t$, denoted as $F_k(S_t,G_t)$, are the vertices whose degrees become at least $k$ due to the selection of the anchored vertex set $S_t$.   
\end{myDef}

\begin{myDef}[\TT{Anchored $k$-core~\cite{DBLP:conf/icalp/BhawalkarKLRS12}}] \label{def:anchored}
Given \TT{an undirected} graph $G_t$ and an anchored vertex set $S_t$, the anchored $k$-core $C_k(S_t)$ consists of the $k$-core of $G_t$, $S_t$, and the followers of $S_t$. 
\end{myDef}

\begin{myExmp}
Consider the graph $G_1$ in Figure~\ref{fig:Motivation}, the $3$-core is $C_3= \{u_8, u_9, u_{12}, u_{13}, u_{16}\}$. If we give users $u_{7}$ and $u_{10}$ a special budget to join in $C_3$, the users $\{u_2, u_3, u_5, u_6,u_{11}\}$ could be brought into $C_3$ because they have no less than $3$ neighbors in $C_3$. Hence, the size of $C_3$ is enlarged from 16 to 23 with the consideration of $u_7$ and $u_{10}$ being the ``anchored'' vertices where the users $\{u_2, u_3, u_5, u_6,u_{11}\}$ are the ``followers" of anchored vertex set $S = \{u_7, u_{10}\}$. Also, the anchored $3$-core of $S$ would be $C_3(S) = \{u_2, u_3, u_5,.., u_{14}, u_{16}\}$.
\end{myExmp}

\begin{algorithm}[t]\footnotesize
  \caption{Core decomposition($G_t,k$)} \label{alg:k-core}
  $k \leftarrow 1$; \\
  \While{$V$ is not empty}{
	\While{exists $u \in V$ with $nbr(u,G_t) < k$}{ \label{line:k-coredecomposition}
	  $V\leftarrow V \setminus \{u\}$; \\
	  $core(u)\leftarrow k-1$; \\
	  \For{$w\in nbr(u,G_t)$}{
	  $nbr(w,G_t)\leftarrow nbr(w,G_t) - 1$;   \\
	  }
	}
	$k \leftarrow k+1$; \\	
  }	
  \Return{$core$;} \\
\end{algorithm}
\setlength{\textfloatsep}{0.2mm} 

\subsection{Problem Statement}

The traditional anchored $k$-core problem aims to explore anchored vertex set for static social networks. However, in real-world social networks, the network topology is almost always evolving over time. Therefore, the anchored vertex set, which maximizes the $k$-core size, should be constantly updated according to the dynamic changes of the social networks. 
In this paper, we model the evolving social network as a series of snapshot graphs $\mathcal{G} = \{G_{t}\}_1^T$. Our goal is to track a series of anchored vertex set $S=\{S_1,S_2,..,S_T\}$ that maximizes the k-core size at each snapshot graph $G_t$ where $t=1,2,..,T$. More formally, we formulate the above task as the \emph{Anchored Vertex Tracking} problem. 

\vspace{2mm}
\noindent
\textbf{Problem formulation: }
Given an \TT{undirected} evolving graph $\mathcal{G} = \{G_t\}^T_1$, the parameter $k$, and an integer $l$, the problem of \textit{anchored vertex tracking (AVT)} in $\mathcal{G}$ aims to discover a series of anchored vertex set $\mathcal{S}=\{S_t\}_1^T$ , satisfying
\begin{equation}
S_t = \arg\max_{|S_t|\leq l} |\mathcal{C}_k(S_t)|
\end{equation}
where $t\in [1,T]$, and $S_t\subseteq V$.

\begin{myExmp}
\TT{In Figure~\ref{fig:Motivation}, if we set $k = 3$ and $l = 2$, the result of the anchored vertex tracking problem can be $\mathcal{S} = \{S_1, S_2, ...\}$ with $S_1 = \{u_7, u_{10}\}$, $S_2 = \{u_7, u_{15}\}$. Besides, the related anchored $k$-core of snapshot graph $G_1$ and $G_2$ would be $\mathcal{C}_k(S_1) = \{u_2,u_3, u_5, u_6,..,u_{13},u_{16}\}$ and $\mathcal{C}_k(S_2) = \{u_2,u_3,u_5,u_6,..,u_{16}\}$, respectively.} 
\end{myExmp}

\section{Problem analysis}\label{sec:analysis}
In this section, we discuss the problem complexity of AVT. In particular, we will verify that the AVT problem can be solved exactly while $k = 1$ and $k = 2$ but become intractable for $k \geq 3$.

\begin{myTheo} \label{theo:problem_complexity} \label{Def:NPhard}
Given an \TT{undirected} evolving general graph $\mathcal{G} = \{G_t\}^T_1$, the problem of AVT is NP-hard when $k \geq 3$. 
\end{myTheo}
\begin{proof}
(1) When $k = 1$ and $t\in [1,T]$, the followers of any selected anchored vertex would be empty. Therefore, we can randomly select $l$ vertices from $\{G_t \setminus C_1\}$ as the anchored vertex set of $G_t$ where $G_t$ is the snapshot graph of $\mathcal{G}$ and $C_1$ is the $1$-core of $G_t$. 
Besides, the time complexity of computing the set of $\{G_t \setminus C_1\}$ from snapshot graph $G_t$ is $\mathcal{O}(|V| +|E_t|)$. 
Thus, the AVT problem is solvable in polynomial time with the time complexity of $\mathcal{O}(\sum_{t=1}^T(|V| + |E_t|))$ while $k = 1$.   

\vspace{1mm}
(2) When $k = 2$ and $t\in [1,T]$, we note that the AVT problem can be solved by repeatedly answering the anchored $2$-core at each snapshot graph $G_t\in \mathcal{G}$. Besides, Bhawalkar et al.~\cite{DBLP:conf/icalp/BhawalkarKLRS12} proposed an exactly \textit{Linear-Time Implementation} algorithm to solve the anchored $2$-core problem in the snapshot graph $G_t$ with time complexity $\mathcal{O}(|E_t| + |V|log|V|)$. From the above, we can conclude that there is an implementation of the algorithm to answer the AVT problem by running in time complexity $\mathcal{O}(\sum_{t=1}^T(|E_t| + |V|log|V|))$. Therefore, the AVT problem is solvable in polynomial time while $k = 2$.   


\vspace{1mm}
(3) When $k \geq 3$ and $t\in [1,T]$, we first note that the anchored vertex tracking problem is equivalent to a set of anchored $k$-core problems at snapshot graphs $G_t\in \mathcal{G}$. 
Thus, we can conclude that the anchored vertex tracking problem is NP-hard once the anchored $k$-core problem is NP-hard.

Next, we prove the problem of anchored $k$-core at each snapshot graph $G_t \in \mathcal{G}$ is NP-hard, by reducing the anchored $k$-core problem to the \textit{Set Cover} problem~\cite{DBLP:conf/coco/Karp72}.
Given a fix instance $l$ of set cover with $s$ sets $S_1, .., S_s$ and $n$ elements $\{e_1,.., e_n\} = \bigcup_{i=1}^s S_i$, we first give the construction only for instance of set cover such that for all $i$, $|S_i|\leq k-1$. 
\TT{In the following}, we construct a corresponding instance of the anchored $k$-core problem in $G_t$ by lifting the above restriction while still obtaining the same results.

Considering $G_t$ contains a set of nodes $V = \{u_1,...,u_n\}$ which is associated with a collection of subsets $\mathcal{S} = \{S_1,..., S_s\}$, $S_i\subseteq V$. We construct an arbitrarily large graph $G'$, where each vertex in $G'$ has degree $k$ except for a single vertex $v(G')$ that has degree $k-1$. Then, we set $H = \{G'_1, ..., G'_m\}$ as the set of $n$ connected components $G'_j$ of $G'$, where $G'_j$ is associated with an element $e_j$. When $e_j \in S_i$, there is an edge between $u_i$ and $v(G'_j)$. 
Based on the definition of $k$-core in Definition~\ref{def:k-core}, once there exists $i$ such that $u_i$ is the neighbor of $v(G'_j)$, then all vertices in $G'_j$ will remain in $k$-core. Therefore, if there exists a set cover $C$ with size $l$, we can set $l$ anchors from $u_i$ while $S_i\in C$ for each $i$, and then all vertices in $H$ will be the member of $k$-core. Since we are assuming that $|S_i|< k$ for all sets, each vertex $u_i$ will not in the subgraph of $k$-core unless $u_i$ is anchored. Thus, we must anchor some vertex adjacent to $v(G'_j)$ for each $G'_j\in G'$,  
which corresponds precisely to a set cover of size $l$.
From the above, we can conclude that for instances of set cover with maximum set size at most $k-1$, there is a set cover of size $l$ if and only if there exists an assignment in the corresponding anchored $k$-core instance using only $l$ anchored vertices such that all vertices in $H$ keep in $k$-core. Hence, the remaining question of reducing the anchored $k$-core problem to the \textit{Set Cover} problem is to lift the restriction on the maximum set size, i.e. $|S_i|\leq k-1$. Bhawalkar et al.~\cite{DBLP:conf/icalp/BhawalkarKLRS12} proposed a $d$-ary tree (defined as $tree(d,y)$) method to lift this restriction. Specifically, to lift the restriction on the maximum set size, they use $tree(k-1, |S_i|)$ to replace each instance of $u_i$. Besides, if $y_1,..., y_{|S_i|}$ are the leaves of the $d$-ary tree, then the pairs of vertices $(y_j, u_j)$ will be constructed for each $u_j\in S_i$.

Since the \textit{Set Cover} problem is NP-hard, we prove that the anchored $k$-core problem is NP-hard for $k\geq 3$, and so is the anchored vertex tracking problem. 
\end{proof}

We then consider the inapproximability of the anchored vertex tracking problem.

\begin{myTheo} \label{theo:approximate}
For $k \geq 3$ and any positive constant $\epsilon > 0$, there does not exist a polynomial time algorithm to find an approximate solution of AVT problem within an $\mathcal{O}(n^{1-\epsilon})$ multiplicative factor of the optimal solution in general graph, unless P = NP.
\end{myTheo}

\begin{proof}
We have reduced the anchored vertex tracking (AVT) problem from the \textit{Set Cover} problem in the proof of Theorem~\ref{Def:NPhard}. Here, we show that this reduction can also prove the inapproximability of AVT problem. For any $\epsilon > 0$, the Set Cover problem cannot be approximated in polynomial time within $(n^{1-\epsilon}) -$ ratio, unless $P = NP$~\cite{DBLP:journals/jacm/Feige98}. Based on the previous reduction in Theorem~\ref{Def:NPhard}, every solution of the AVT problem in the instance graph $G$ corresponds to a solution of the \textit{Set Cover} problem. 
Therefore, it is NP-hard to approximate anchored vertex tracking problem on general graphs within a ratio of $(n^{1-\epsilon})$ when $k \geq 3$. 
\end{proof}

\section{The Greedy Algorithm} \label{sec:alg}
Considering the NP-hardness and inapproximability of the AVT problem, we first resort to developing a Greedy algorithm to solve the AVT problem. Algorithm~\ref{alg:basic1} summzrizes the major steps of the Greedy algorithm. 
The core idea of our Greedy algorithm is to iteratively find the $l$ number of best anchored vertices which have the largest number of followers in each snapshot graph $G_t\in \mathcal{G}$  (Lines~\ref{R3:alg2:0}-\ref{R3:alg2:4}). For each $G_t\in \mathcal{G}$ where $t$ is in the range of $[1,T]$ (Line~\ref{R3:alg2:0}), in order to find the best anchored vertex in each of the $l$ iterations (Lines~\ref{R3:alg2:1}), we compute the followers of every candidate anchored vertex by using the \textit{core decomposition} process mentioned in Algorithm~\ref{alg:k-core} (Lines~\ref{R3:alg2:2}-\ref{R3:alg2:3}). 
Specifically, considering the $k$-core $C_k$ of $G_t$, if a vertex $u$ is anchored, then the core decomposition process repeatedly deletes all vertices (except $u$) of $G_t$ with the degree less than $k$. Thus, the remaining vertices that do not belong to $C_k$ will be the followers of $u$ with regard to the $k$-core. In other words, these followers will become the new $k$-core members due to the anchored vertex selection.  
From the above process of the Greedy algorithm, we can see that every vertex will be the candidate anchored vertex in each snapshot graph $G_t = (V, E_t)$, and every edge will be accessed in the graph during the process of core decomposition. Hence, the time complexity of the Greedy algorithm is $\mathcal{O}(\sum_{t=1}^T l\cdot |V| \cdot |E_t|)$.


\begin{algorithm}[!t] \footnotesize
	\caption{\textit{The Greedy Algorithm}} \label{alg:basic1}
	\KwIn{$\mathcal{G}=\{G_t\}^T_1:$ an evolving graph, $l$: the allocated size of anchored vertex set, and $k$: degree constraint}
	\KwOut{ $\mathcal{S}=\{S_t\}_1^T:$ the series of anchored vertex sets} 
	\BlankLine
	$\mathcal{S}\leftarrow \emptyset$; \\ 
	\For {each $t\in [1,T]$}{ \label{R3:alg2:0}
	    $i\leftarrow 0$; $S_t\leftarrow \emptyset$ \\
	    \While{$i< l$}{ \label{R3:alg2:1}
	          \textcolor{black}{/* \ \ Candidate Anchored Vertex \ \ \ \ \ */} \\
              \For{each $u\in V$}{ \label{R3:alg2:2}
              \textcolor{black}{/*\ \ Computing Followers \ \ \ \ \ \ \ \ \ \ */ } \\
                  \textcolor{black}{ Compute $F_k(u, G_t)$}; \\ \label{R3:alg2:3}
              }
              $u'\leftarrow$ the best anchored vertex in this iteration; \\
              $S_t\leftarrow S_t\cup u'$; $i\leftarrow i+1$; \\ 
         }
       $\mathcal{S}\leftarrow \mathcal{S}\cup S_t$; \\ \label{R3:alg2:4}
     }
\Return{$\mathcal{S}$}	
\end{algorithm}

Since the Greedy algorithm's time complexity is cost-prohibitive, we need to accelerate this algorithm from two aspects: (i) reducing the number of potential anchored vertices; and (ii) accelerating the followers' computation with a given anchored vertex.

\subsection{Reducing Potential Anchored Vertices} \label{subsec:Base2}
In order to reduce the potential anchored vertices, we present the below definition and theorem to identify the quality anchored vertex candidates. 

\begin{myDef}[$K$-order \cite{DBLP:conf/icde/ZhangYZQ17}]\label{K-order}
	Given two vertices $u, v\in V$, the relationship $\preceq$ in $K$-order index holds $u \preceq v$ in either $core(u) < core(v)$; or $core(u) = core(v)$ and $u$ is removed before $v$ in the process of core decomposition. 
\end{myDef}

\begin{figure}[h]
	\centering
		\includegraphics[scale=0.5]{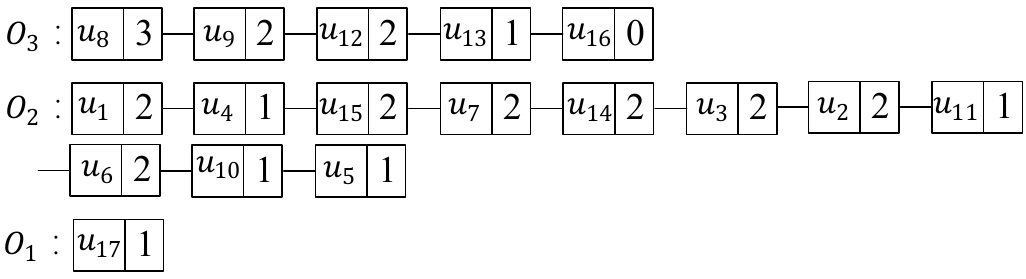}
		\vspace{-2mm}
		\caption{The $K$-order $O$ of graph $G_1$ in Figure~\ref{fig:Motivation}}
	\label{fig:K-order}
	\vspace{-2mm}
\end{figure}

Figure~\ref{fig:K-order} shows a $K$-order index $O = \{\mathcal{O}_1, \mathcal{O}_2, \mathcal{O}_3\}$ of graph snapshot $G_1$ in Figure~\ref{fig:Motivation}. The vertex sequence $\mathcal{O}_k \in O$ records all vertices in $k$-core by following the removing order of core decomposition, \textit{i.e., $\mathcal{O}_2$ records all vertices in 2-core and vertex $u_1$ is removed early than vertex $u_4$ during the process of core decomposition in $G_1$}. 

\begin{myTheo} \label{R3:theo1}
Given a graph snapshot $G_t$, a vertex $x$ can become an anchored vertex candidate if $x$ has at least one neighbor vertex $v$ in $G_t$ that satisfies: the neighbor vertex's core number must be k-1 (i.e., $core(v)=k-1$), and $x$ is positioned before the neighbor node $v$ in K-order (i.e., $x\preceq v$).
\end{myTheo}

\begin{proof}
We prove the correctness of this theorem by contradiction. If $v\preceq x$ in the $K$-order of $G_t$, then $v$ will be deleted prior to $x$ in the process of core decomposition in Algorithm~\ref{alg:k-core}. In other words, anchoring $x$ will not influence the core number of $v$. Therefore, $v$ is not the follower of $x$ when $v\preceq x$. On the other hand, it is already proved in~\cite{DBLP:journals/pvldb/ZhangZZQL17} that only vertices with core number $k-1$ may be the follower of an anchored vertex. If no neighbor of vertex $x$ has core number $k-1$, then anchoring $x$ will not bring any followers, which is contradicted with the definition of the anchored vertex.
From above analysis, we can conclude that the candidate anchored vertex only comes from the vertex $x$ which has at least one neighbor $v$ with core number $k-1$ and behind $x$ in $K$-order, i.e., $\{x\in V|\exists v\in nbr(x,G_t)\wedge core(v)=k-1\wedge x\preceq v\}$. Hence, the theorem is proved. 
\end{proof}

According to Theorem~\ref{R3:theo1}, the anchored vertex candidates will be probed only from the vertices that can bring some followers into the k-core. This also meets the requirement of anchored k-core in Definition~\ref{def:anchored}. Thus, the size of potential anchored vertices at each snapshot graph $G_t$ can be significantly reduced from $|V|$ to $|\{x\in V|\exists v\in nbr(x,G_t)\wedge core(v)=k-1\wedge x\preceq v\}|$.

\begin{myExmp}
Given the graph $G_1$ in Figure~\ref{fig:Motivation} and $k =3$, $u_{15}$ can be selected as an anchored vertex candidate because
anchoring $u_{15}$ would bring the set of followers, $\{u_{14}\}$, into the anchored 3-core. 
\end{myExmp}

\subsection{\TT{Accelerating} Followers Computation} \label{sec:greedy_acc}
To accelerate the computation of followers, a feasible way is to transform the followers' computation into the \textbf{core maintenance} problem~\cite{DBLP:journals/tkde/LiYM14,DBLP:conf/icde/ZhangYZQ17}, which 
aims to maintain the core number of vertices in a graph when the graph changes. The above problem transformation is based on an observation: 
given an anchored vertex $u$, its followers' core number can be increased to $k$ value if $core(u)$ is treated as infinite according to the concept of anchored node. 

Therefore, we modify the state-of-the-art core maintenance algorithm, \textit{OrderInsert}~\cite{DBLP:conf/icde/ZhangYZQ17}, to compute the followers of an anchored vertex $u$ in snapshot graph $G_t$. Explicitly, we first build the $K$-order of $G_t$ using \textit{core decomposition} method described in Algorithm~\ref{alg:k-core}. For each anchored vertex candidate $u$, we set the core number of $u$ as infinite and denote the set of its followers as $V^*$ initialized to be empty. 
After that, we iteratively update the core number of $u$'s neighbours and other affected vertices by using the \textit{OrderInsert} algorithm, and record the vertices with core number increasing to $k$ in $V^*$. 
Finally, we output $V^*$ as the follower set of $u$. 

Besides, we introduce two notations, \textbf{remaining degree} (denoted as $deg^+()$) and \textbf{candidate degree} (denoted as $deg^-()$), to depict more details of the above followers' computation method. Specifically, for a vertex $u$ in snapshot graph $G_t$ where $core(u) = k-1$, $deg^+(u)$ is the number of remaining neighbors when $u$ is removing during the process of core decomposition, \textit{i.e., $deg^+(u)= |{v\in nbr(u,G_t): u \preceq v}|$}. And $deg^-(u)$ records the number of $u$'s neighbors $v$ included in $\mathcal{O}_{k-1}$ but appearing before $u$ in $\mathcal{O}_{k-1}$, and $v$ is in followers set $V^*$, \textit{i.e., $deg^-(u) = |\{v \in nbr(u,G_t): v\preceq u \wedge core(v) = k-1\wedge v\in V^*\}|$}. Since, $deg^+(u)$ records the number of $u$'s neighbors after $u$ in the $K$-order having core numbers larger than or equal to $k-1$, $deg^+(u) + deg^-(u)$ is the upper bound of $u$'s neighbors in the new $k$-core. Therefore, all vertices $s$ in follower set $V^*$ must have $deg^+(s) + deg^-(s)\geq k$.  

\begin{algorithm}[!t] \footnotesize
	\caption{\textit{ComputeFollower}($G_t$, $u$, $\mathcal{O}(G_t)$) } \label{alg:basic2}
	\BlankLine
    
    $K$-order $O(G_t) = \{\mathcal{O}_1, \mathcal{O}_2,..., \mathcal{O}_{max}\}$ \\ \label{alg2:inp}
    $F_k(u, G_t)\leftarrow \emptyset$; \\  \label{alg2:0}
	\For{$v\in nbr(u,G_t)$}{ \label{alg2:1}
	    $deg^-(.)\leftarrow 0$; $V^*\leftarrow \emptyset$; \\ 
	    \textcolor{black}{/* Core phase of the \textit{OrderInsert} algorithm~\cite{DBLP:conf/icde/ZhangYZQ17} */} \\
	    \If{$core(v)= k-1 \ \& \ u\preceq v$}{  \label{alg2:2}
	       $deg^+(v)\leftarrow deg^+(v)+1$ \\
	       \If{$deg^+(v)+ deg^-(v)> k-1$}{
              remove $v$ from $\mathcal{O}_{k-1}$ and append it to $V^*$; \\
              \For{ $w\in nbr(v) \wedge w\in \mathcal{O}_{k-1} \wedge v \preceq w $}{
                  $deg^-(w) \leftarrow deg^-(w) + 1$; \\
            }
              Visit the vertex next to $v$ in $\mathcal{O}_{k-1}$; \\  	       
	       }
	       \Else{
	       \If{$deg^-(v) = 0$} {
	           Visit the vertex next to $v$ in $\mathcal{O}_{k-1}$; \\  	
	       }
	       \Else{
	          \For{each $w\in nbr(v) \wedge w\in V^*$}{
	              \If{$deg^+(w) + deg^-(w) < k$}{
	                  remove $w$ from $V^*$; \\
	                  update $deg^+(w)$ and $deg^-(w)$; \\
	                  $nbr(v) \leftarrow nbr(v) \cup nbr(w)$; \\
	                  Insert $w$ next to $v$ in $\mathcal{O}_{k-1}$; \\ 
	              }
	          }
              Visit the vertex with $deg^-(.) = 0$ and next to $v$ in $\mathcal{O}_{k-1}$; \\
	       }
	       }
	    }
	    \Else{
           Continue; \\	   \label{alg2:3}   
	    }
	    Insert vertices in $V^*$ to the beginning of $\mathcal{O}_k$ in $O(G_t)$;\label{alg2:5} \\    
	    $F_k(u, G_t)\leftarrow F_k(u, G_t)\cup V^*$; \\  \label{alg2:4}
	}
\Return{$F_k(u, G_t)$} \label{alg2:6}	
\end{algorithm}
\setlength{\textfloatsep}{0.2mm} 

The pseudocode of the above process is shown in Algorithm~\ref{alg:basic2}. Initially, the $K$-order of $G_t$ is represented as $O(G_t) = \{\mathcal{O}_1, \mathcal{O}_2,..., \mathcal{O}_{max}\}$ where $max$ represents the maximum core number of vertices in $G_t$ (Line~\ref{alg2:inp}). We then set the followers set of anchored vertex $u$, $F_k(u, G_t)$ as empty (Line~\ref{alg2:0}). For each $u$'s neighbours $v$ (Line~\ref{alg2:1}), we iteratively using the  \textit{OrderInsert} algorithm~\cite{DBLP:conf/icde/ZhangYZQ17} to update the core number of $v$ and the other affected vertices due to the core number changes of $v$, and record the vertices with core number increasing to $k$ in a set $V^*$ (Lines~\ref{alg2:2}-\ref{alg2:5}).    
After that, we add $V^*$ related to each $u$'s neighbors $v$ into $u$'s follower set $F_k(u,G_t)$ (Line~\ref{alg2:4}). Finally, we output $F_k(u, G_t)$ as the followers set of $u$ (Line~\ref{alg2:6}).  

\begin{myExmp}
Using Figure~\ref{fig:K-order} and Figure~\ref{fig:Motivation}, we would like to show the process of followers' computation. Assume $k = 3$, $V^* = \emptyset$, and the $K$-order, $O=\{\mathcal{O}_1, \mathcal{O}_2, \mathcal{O}_3\}$, in graph $G_1$. Initially, the $deg^+(u)$ value of each vertex $u$ is recorded in $O(G_1)$, i.e., $deg^+(u_{14}) = 2$, $deg^-() = 0$ for all vertices in $G_1$ as $V^*$ is empty. If we anchor the vertex $u_{15}$, i.e., $core(u_{15}) = \infty$, then we need to 
update the candidate degree value of $u_{15}$'s neighbours in $\mathcal{O}_2$, i.e., $deg^-(u_{11}) = 0 + 1$ and $deg^-(u_{14}) = 0 + 1$. We then start to visit the foremost neighbours of $u_{15}$ in $\mathcal{O}_2$, i.e., $u_{14}$. Since $deg^+(u_{14}) + deg^-(u_{14}) = 2 +1 \geq 3$ and $deg^+(u_{11}) + deg^-(u_{11}) = 1 +1 < 3$, we can add $u_{14}$ in $V^*$ and then update the $deg^-()$ of its impacted neighbours. After that, we 
sequentially explore the vertices $s$ after $u_{14}$ in $\mathcal{O}_2$, and operate the above steps once $deg^+(s) + deg^-(s) \geq 3$. The follower computation terminates when the last vertex in $\mathcal{O}_2$ is processed, i.e., $u_{11}$. Therefore, the $V^*$ related to $u_{14}$ is $\{u_{14}\}$, and the follower set of $u_{15}$ is $F_k(u_{15},G_1) = \emptyset \cup V^* = \{u_{14}\}$. Finally, we output the follower set of $u_{15}$, i.e., $F_k(u_{15},G_1) = \{u_{14}\}$.  
\end{myExmp}

The time complexity of Algorithm~\ref{alg:basic2} is calculated as follows. 
The followers' computation of an anchored vertex $u$ can be transformed as the core maintenance problem under inserting edges $(u,v)$ where $v$ is the neighbor of $u$. Meanwhile, Zhang et al.~\cite{DBLP:conf/icde/ZhangYZQ17} reported that the core maintenance process while inserting an edge takes $\mathcal{O}(\sum_{v\in V^+}deg(v)\cdot logmax\{|C_{k-1}|, |C_k|\})$ (Lines~\ref{alg2:2}-\ref{alg2:5}), and $V^+$ is a small set with average size less than $3$. Therefore, we conclude that the time complexity of Algorithm~\ref{alg:basic2} is $\mathcal{O}(\sum_{v\in nbr(u)}\sum_{v\in V^+}deg(v)\cdot logmax\{|\mathcal{O}_{k-1}|, |\mathcal{O}_k|\})$. The time complexity of the above followers' computation method is far less than directly using \textit{core decomposition} to compute the followers of a given anchored vertex.  

\section{Incremental Computation Algorithm} \label{sec:improved}
For an evolving graph $\mathcal{G}$, the Greedy approach individually constructs the $ K $-order and iteratively searches the anchored vertex set at each snapshot graph $G_t$ of $\mathcal{G}$. \TT{However,} it does not fully exploit the connection of two neighboring snapshots to advance the performance of solving AVT problem. 
To address the limitation, in this section, we propose a bounded $K$-order maintenance approach that can avoid the reconstruction of the $K$-order at each snapshot graph. With the support of our designed $K$-order maintenance, we develop an incremental algorithm, called \textit{IncAVT}, to find the best anchored vertex set at each graph snapshot more efficiently. 

\subsection{The Incremental Algorithm Overview}
Let $\mathcal{G}=\{G_1,G_2,..,G_T\}$ be an evolving graph, $\mathcal{S}_t$ be the anchored vertex result set of AVT in $G_t$ where $t\in [1,T]$. $E^+$ and $E^-$ represent the number of edges to be inserted and deleted at the time when $G_{t-1}$ evolves to $G_t$.  
To find out the anchored vertex sets $\mathcal{S} = \{S_t\}^T_1$ of $\mathcal{G}$ using the \textit{IncAVT} algorithm, we first build the $K$-order of $G_1$, and then compute the anchored vertex set $S_1$ of $G_1$. 
Next, we develop a bounded $K$-order maintenance approach to maintain the $K$-order by considering the change of edges from $G_{t-1}$ to $G_t$. The benefit of this approach is to avoid the $K$-order reconstruction at each snapshot $G_t$.
Meanwhile, during the process of $K$-order maintenance, we use vertex sets $V_I$ and $V_R$ to record the vertices that are impacted by the edge insertions and edge deletions, respectively.
After that, we iteratively find the $l$ number of best anchored vertices in each snapshot graph $G_t$, while the potential anchored vertices are selected to probe from $V_I$, $V_R$, and $S_{t-1}$. The $l$ anchored vertices are recorded in $\mathcal{S}_t$. Finally, we output $\mathcal{S} =\{S_t\}^T_1$ as the result of the AVT problem.

\subsection{Bounded $K$-order Maintenance Approach}
In this subsection, we devise a bounded $K$-order maintenance approach to maintain the $K$-order while the graph evolving from $G_{t-1}$ to $G_t$, i.e., $t\in [2, T]$. Our bounded $K$-order maintenance approach consists of two components: (1) \textbf{EdgeInsert}, handling the $K$-order maintenance while inserting the edges $E^+$; and (2) \textbf{EdgeRemove}, handling the $K$-order maintenance while deleting the edges $E^-$. 

\subsubsection{Handling Edge Insertion}

If we insert the edges in $E^+$ into $G_{t-1}$, then the core number of each vertex in $G_{t-1}$ either remains unchanged or increases. Therefore, the $k$-core of snapshot graph $G_{t-1}$ is part of the $k$-core of snapshot graph $G_{t}$ where $G_t = G_{t-1} \oplus E^+$. 
The following lemmas show the update strategies of core numbers of vertices when the edges are added.

\begin{myLem} \label{R3:lem1}
Given a new edge $(u,v)$ that is added into $G_{t-1}$, the remaining degree of $u$ increases by 1, i.e.,  $deg^+(u) = deg^+(u) + 1$, if $u\preceq v$ holds. 
\end{myLem}

\begin{proof}
From Section~\ref{sec:greedy_acc} of the remaining degree of a vertex, we get  $deg^+(u) = |\{v\in nbr(u)\ |\ u\preceq v\}|$. Inserting an edge $(u,v)$ into graph snapshot $G_{t-1}$ brings one new neighbour $v$ to $u$ where $u\preceq v$ in the $K$-order of $G_{t-1}$, i.e., $O(G_{t-1})$. Therefore, $deg^+(u)$ needs to increase by $1$ after inserting $(u,v)$ into $G_{t-1}$. 
\end{proof}

\begin{myExmp}
Consider the snapshot graph $G_1$ in Figure~\ref{fig:Motivation}, if we add a new edge $(u_2, u_5)$ into $G_1$ where $u_2\preceq u_5$ (mentioned in Figure~\ref{fig:K-order}), then the remaining degree of $u_2$, $deg^+(u_2) = deg^+(u_2) +1 = 3$.
\end{myExmp}

\begin{myLem} \label{R3:lem2}



Let $deg^+(u)$ and $core(u)$ be the remaining degree and core number of vertex $u$ in snapshot graph $G_t$ respectively. Suppose we insert a new edge $(u,v)$ into $G_t$ and update $deg^+(u)$. Thus, the core number $core(u)$ of $u$ may increase by 1 if $core(u) < deg^+(u)$. Otherwise, $core(u)$ remains unchanged. 
\end{myLem}

\begin{proof}
We prove the correctness of this lemma by contradiction. 
From Definition~\ref{def:corenumber} and the definition of remaining degree in Section~\ref{sec:greedy_acc}, we know that if $u$'s core number \TT{does not need to be updated} after inserting edge $(u,v)$ into $G_{t-1}$, then the number of $u$'s neighbours $v$ with $u\preceq v$ must be no more than $core(u)$. Therefore, the value of updated $deg^+(u)$ should be no more than $core(u)$, which is contradicted with the fact that $core(u)< deg^+(u)$.  
\end{proof}

\begin{myExmp}
Considering a vertex $u_2$ in graph $G_1$, we can see $deg^+(u_2) = 2$, and $core(u_2) = 2$ as shown in Figure~\ref{fig:Motivation} and Figure~\ref{fig:K-order}. If an edge $(u_2, u_5)$ is inserted into $G_1$, we can get $deg^+(u_2) = deg^+(u_2) + 1 =3$ (refer Lemma~\ref{R3:lem1}). Since $core(u_2) = 2 < deg^+(u_2) = 3$, the $core(u_2)$ may increase by 1 according to Lemma~\ref{R3:lem2}.
\end{myExmp}

We present the \textit{EdgeInsert} algorithm for $K$-order maintenance. It consists of three main steps. Firstly, for each vertex $u$ relating to the inserting edges $(u,v)\in E^+$, we need to update its remaining degree, \textit{i.e., $deg^+(u)$} (refer Lemma~\ref{R3:lem1}). 
Then, we identify the vertices impacted by the insertion of $E^+$ and update its remaining degree value, core number, and positions in $K$-order (refer Lemma~\ref{R3:lem2}). This step is the core phase of our algorithm. 
Finally, we add the vertex $u$ into the vertex set $V_I$ if $u$ 
has the updated core number $core(u) = k-1$ after inserting $E^+$. This is because the followers only come from vertices with core number $k-1$ (refer Theorem~\ref{R3:theo1}).   

\begin{algorithm}[t] \footnotesize
\small
	\caption{\textit{EdgeInsert}$(G'_{t}$, $O$, $E^+$, $k$)} \label{alg:insert}
	\BlankLine
	\TT{$i\leftarrow 0$}, $V_I\leftarrow \emptyset$, $m \leftarrow 0$, $O'\leftarrow \emptyset$; \\
	\For{each $e=(u,v) \ \& \ e\in E^+$}{ \label{R3:alginsert_1}
	    $m \leftarrow \max\{m, \min(core(u), core(v))\}$; \\
	    $u \preceq v$ ? $deg^+(u) + = 1 : deg^+(v) + = 1$; \\ \label{R3:alginsert_2}
	}
	\While{$i\leq m$}{ \label{R3:alginsert_3}
	  $V_C\leftarrow \emptyset$, $deg^-(.)\leftarrow 0$;\\  \label{R3:alginsert_5}
      $u^* \leftarrow$ the first vertex of $\mathcal{O}_i\in O$; \\ \label{R3:alginsert_6}
      \While{$u^*\neq nil$}{ \label{R3:alginsert_7}
           \If{$deg^+(u^*) + deg^-(u^*)> i$}{ \label{R3:alginsert_9}  
              remove $u^*$ from $\mathcal{O}_i$; 
              append $u^*$ into $V_C$;\\
              \If{$i = k-1$}{
                  add $u^*$ into $V_I$              
              }
              \For{each $v\in nbr(u^*,G'_t)\wedge core(v)=i \wedge u^*\preceq v$}{\
                  $deg^-(v)\leftarrow deg^-(v) + 1$; \\\label{R3:alginsert_10}  
              }           
           }
           \Else{
               \If{$deg^-(u^*)=0$}{   \label{R3:alginsert_11}  
                    remove $u^*$ from $\mathcal{O}_i$; 
                    append $u^*$ to $\mathcal{O}_{i'}$; \\  \label{R3:alginsert_12}  
               }
               \Else{ 
                    $deg^+(u^*)\leftarrow deg^+(u^*) + deg^-(u^*)$; \\\label{R3:alginsert_13} 
                    $deg^-(u^*)\leftarrow 0$; \\
                    remove $u^*$ from $\mathcal{O}_i$;         
                    append $u^*$ to $\mathcal{O}_{i'}$; \\   
                    update the $deg^+(.)$ of $u^*$'s neighbors; \\ \label{R3:alginsert_8}
               } 
           }
           $u^*\leftarrow$ the vertex next to $u^*$ in $\mathcal{O}_i$; \\  
      }
      \For {$v\in V_C$}{ \label{R3:alginsert_14}  
          $deg^-(v)\leftarrow 0$; 
          $core(v)\leftarrow core(v)+1$; \\   
          \If {$i = k-1$}{
              remove $v$ from $V_I$; \\          
          }  
      } 
      insert vertex set $V_C$ into the beginning of $\mathcal{O}_{i+1}$; \\
      \If {$i=k-2$}{
          $V_I\leftarrow V_I \cup V_C$; \\   \label{R3:alginsert_15}      
      }
      add $\mathcal{O}_{i'}$ to new $K$-order $O'$ in $G'_t$; \\ 
      $i \leftarrow i+1$; \\	\label{R3:alginsert_4}
	}
\Return {the K-order $O'$ in $G'_t$, and $V_I$}  \label{R3:edgeinsert_final}
\end{algorithm} 
\setlength{\textfloatsep}{0.2mm} 

The detailed description of our \textit{EdgeInsert} algorithm is outlined in Algorithm~\ref{alg:insert}. The inputs of the algorithm are snapshot graph $G_{t-1}$ where $t\in [2,T]$, the $K$-order $O = \{\mathcal{O}_1, \mathcal{O}_2,.., \mathcal{O}_k,..\}$ of $G_{t-1}$, the edge insertion $E^+$, and a positive integer $k$. Initially, for each inserted edge $(u, v)\in E^+$, we increase the remaining degree of $u$ by 1 where vertex $u\preceq v$ (refer Lemma~\ref{R3:lem1}), use $m$ to record the maximum core number of all vertices related to $E^+$ (Lines~\ref{R3:alginsert_1}-\ref{R3:alginsert_2}). 
Next, for $i\in [0, m]$, we iteratively identify the vertices in $\mathcal{O}_i\in O$ whose 
core number increases after the insertion of $E^+$, and we also update $\mathcal{O}_i$ of $K$-order (Lines~\ref{R3:alginsert_3}-\ref{R3:alginsert_4}). Here, a new set $V_C$ is initialized as empty and it will be used to maintain the new vertices whose core number increases from $i-1$ to $i$. And then, we start to select the first vertex $u^*$ from  $\mathcal{O}_i$ 
(Line~\ref{R3:alginsert_6}). In the inner \textit{while} loop, we visit the vertices in $\mathcal{O}_i$ in order (Lines~\ref{R3:alginsert_7}-\ref{R3:alginsert_8}). The visited vertex $u^*$ must satisfy one of the three conditions: (1) $deg^+(u^*) + deg^-(u^*) > i$; (2) $deg^+(u^*) + deg^-(u^*) \leq i \ \wedge \ deg^- (u^*) = 0$; (3) $deg^+(u^*) + deg^-(u^*) \leq i \ \wedge \ deg^- (u^*) > 0$. For condition (1), the core number of the visited vertex $u^*$ may increase. Then, we remove $u^*$ from $\mathcal{O}_i$ and add it into $V_C$. Besides, the candidate degree of each neighbour $v$ of $u^*$ should increase by $1$ if $u^*\preceq v$ (Lines~\ref{R3:alginsert_9}-\ref{R3:alginsert_10}). For condition (2), the core number of $u^*$ will not change. So we remove $u^*$ from the previous $\mathcal{O}_i$ and append it into $\mathcal{O}_{i'}$ of the new $K$-order $O'$ of graph $G'_t = G_{t-1}\oplus E^+$ (Lines~\ref{R3:alginsert_11}-\ref{R3:alginsert_12}). 
For condition (3), we can identify that $u^*$'s core number will not increase. So we need to update the remaining degree and candidate degree of $u^*$, and remove $u^*$ from $\mathcal{O}_i$ and append it to $\mathcal{O}_{i'}$. We also need to update the remaining degree of the neighbours of $u^*$ (Lines~\ref{R3:alginsert_13}-\ref{R3:alginsert_8}). After that, $V_I$ maintains the vertices that are affected by the edge insertion, and these vertices have core number $k-1$ in new $K$-order $O'$ of graph $G'_t$ (Lines~\ref{R3:alginsert_14}-\ref{R3:alginsert_15}). Finally, when the outer \textit{while} loop terminates, we can output the maintained $K$-order and the affected vertices set $V_I$ (Line~\ref{R3:edgeinsert_final}).

\subsubsection{Handling Edge Deletion}
Here, we present the procedure of $K$-order maintenance for edge deletions. The following definitions and lemmas show the update strategies of core numbers of vertices when the edges are deleted. 

\begin{myLem} \label{lem:remove:deg+}
Suppose an edge $(u,v)$ is deleted while graph evolves from $G_{t-1}$ to $G_t$, then the remaining degree of $u$ from $G_{t-1}$ to $G_t$ decreases by $1$, i.e., $deg^+(u) = deg^+(u) - 1$, if $u \preceq v$ holds.
\end{myLem}

\begin{proof}
From Section~\ref{sec:greedy_acc} of the remaining degree of a vertex, we get  $deg^+(u) = |\{v\in nbr(u)\ |\ u\preceq v\}|$. Deleting an edge $(u,v)$ from graph snapshot $G_{t-1}$ evolving to $G_t$ removes one neighbour $v$ of $u$ where $u\preceq v$ in the $K$-order of $G_t$. Therefore, $deg^+(u)$ needs to decrease by $1$ after deleting $(u,v)$ from $G_{t-1}$. 
\end{proof}

\begin{myExmp}
Consider the snapshot graph $G_1$ and $G_2$ in Figure~\ref{fig:Motivation}, if we remove edge $(u_2, u_{11})$ from $G_1$ to $G_2$ where $u_2 \preceq u_{11}$ (mentioned in Figure~\ref{fig:K-order}), then the remaining degree of $u_2$ will decrease from $2$ to $1$. 
\end{myExmp}

We then introduce an important notion, called \textit{max core degree}, and the related lemma. 

\begin{myDef}[Max core degree~\cite{DBLP:journals/pvldb/SariyuceGJWC13}] \label{def:mcd}
Given \TT{an undirected} graph $G_t$, the max-core degree of a vertex $u$ in $G_t$, denoted as $mcd(u)$, is the number of $u$'s neighbours whose core number no less than $core(u)$. 
\end{myDef}

\begin{myExmp}
Consider the snapshot graph $G_1$ in Figure~\ref{fig:Motivation}, we have $core(u_9) = 3$, $core(u_{14}) = 2$, $core(u_{15}) = 2$, $core(u_{16}) = 3$, and $core(u_{17}) = 1$. Therefore, the max core degree of vertex $u_{14}$ is $3$ due to $3$ of $u_{14}$'s neighbors $\{u_9, u_{15}, u_{16}\}$ has core number no less than $core(u_{14})$. 
\end{myExmp}

Based on $k$-core definition (refer Definition~\ref{def:k-core}), $mcd(u) < core(u)$ means that $u$ does not have enough neighbors who meet the requirement of $k$-core. Thus, $u$ itself cannot stay in $k$-core as well. 
Therefore, it can conclude that for a vertex, its max core degree is always larger than or equal to its core number, \textit{i.e, $mcd(u) \geq core(u)$}.



\begin{myLem} \label{R3:mcd}
Let $mcd(u)$ and $core(u)$ be the Max-core degree and core number of vertex $u$ in snapshot graph $G_t$. Suppose we delete an edge $(u,v)$ from $G_t$ and the updated $mcd(u)$. Thus, the core number $core(u)$ of $u$ may decrease by 1 if $mcd(u) < core(u)$. Otherwise, $core(u)$ remain unchanged. 
\end{myLem}

\begin{proof}
Based on Definition~\ref{def:k-core} and Definition~\ref{def:corenumber}, the core number of vertex $u$ is identified by the number of its neighbours with core number no less than $u$. Moreover, a vertex $u$ must have at least $core(u)$ number of neighbours with core number no less than $core(u)$. From Definition~\ref{def:mcd}, the max core degree of a vertex $u$ is the number of $u$'s neighbour with core number no less than $u$, i.e, $mcd(u)=|\{v\ |\ v = nbr(u) \wedge core(v)\geq core(u)\}|$. Therefore, we can conclude that $mcd(u)\geq core(u)$ always holds. Hence, if $mcd(u) < core(u)$ after \TT{deleting} an edge from $G_t$ and \TT{updating} $mcd(u)$, then $core(u)$ also needs to be decreased by $1$ to ensure $mcd(u) > core(u)$ in the changed graph.   
\end{proof}

\begin{algorithm}[t] \footnotesize
\small
	\caption{\textit{EdgeRemove}($G'_t$, $O'$, $E^-$, $k$)} \label{R3:alg:remove}
	\BlankLine
    \textcolor{black}{/* $mcd(u)$ is the number of $u$'s neighbour $v$ with $core(u)\leq core(v)$ */} \\ 
    $O' = \{\mathcal{O}_1, \mathcal{O}_2, ...\}$; \TT{Initialize array $F[|V|]$}\\

    $V_R\leftarrow \emptyset$, and $m\leftarrow 0$; \\ 
    let $Q$ be an empty queue and $V^* = \{V_1, V_2,..\}$, $V_i\in V^*$ be the empty list; \\
   \textcolor{black}{/* identify the vertex need to remove from $O'$ */} \\
    \For{each $e=(u,v) \ \& \ e\in E^-$}{  \label{R3:alg5:1} 
         $u'\leftarrow $ $u$ if $u\preceq v$, otherwise $v$; \\
        \TT{$G_t := G'_t \oplus e$;} 
        $\TT{j} \leftarrow core(u',G'_t)$; \\
         compute $mcd(u', G_t)$ of $u'$;    \label{alg5_x3}    \\
         
         \If{$mcd(u', G_t)< \TT{j}$}{  \label{alg5_x1}
             remove $u'$ from $O'_i$,
             enqueue $u'$ to $Q$; \\ 
             $core(u')\leftarrow core(u')-1$; \\ 
            \TT{ \If{$F[u'] == 1$}{
              remove $u'$ from $V_j$; \\
             }
             \Else{
                $F[u'] == 1$; \\}   \label{R3:alg5:2}
             }
         }  
    
    \While{$Q$ is not empty}{ \label{R3:alg5:3} 
          dequeue $u$ from $Q$, 
          $i \leftarrow core(u,G_t)$; \\
          append $u$ to $V_i$, 
          $m\leftarrow max\{m, i\}$; \\
          \For{$u'\in nbr(u, G'_t)\wedge core(u')\TT{==} j$}{
               repeat lines~\ref{alg5_x3}-\ref{R3:alg5:2}; \\  \label{R3:alg5:4}        
          }   
    }
    \TT{$G'_t := G_t$}; \\
   }
    \textcolor{black}{/* update the k-order $O'$ */} \\
    \For{$i\leftarrow$ $m$ to $1$}{    \label{R3:alg5:5} 
         \For{each $u\in V_i$ in order}{
             $deg^+(u) \leftarrow 0$; \\
             \For {$u'\in nbr(u, G_t)$}{
                \If{$core(u') > core(u) \vee u'\in V_i$}{
                    $deg^+(u)\leftarrow deg^+(u) + 1$; \\
                }
                recompute $deg^+(u')$; \\
             }
             append $u$ to the end of $\mathcal{O}_{i}$; \\    \label{R3:alg5:6}     
         }          
    }
    $V_R\leftarrow V_{k-1}$,   
    $O(G_t) \leftarrow O'$;\\ \label{R3:alg5:7} 
\Return {the $K$-order $O(G_t)$ of $G_t$, and $V_R$} 
\end{algorithm}  
\setlength{\textfloatsep}{0.2mm}


The \textit{EdgeRemove} algorithm is presented in Algorithm~\ref{R3:alg:remove}. The inputs of the algorithm are the graph $G'_t$ constructed by $G_{t-1}$ with the insertion edges of $E^+$, i.e., $G'_t = G_{t-1} \oplus E^+$, and $O'$ is the $K$-order of $G'_t$. The main body of Algorithm~\ref{R3:alg:remove} consists of three steps. In the first step (Lines~\ref{R3:alg5:1}-\ref{R3:alg5:4}), we identify the vertices that needs to be removed from their previous position of $K$-order $O'$ after the edge deletion. Specifically, we first update the graph $G_t$, 
and then compute the max core degree of these vertices (Line~\ref{alg5_x3}). Meanwhile, we add the influenced vertex $u$ related to the deleting edges, i.e., $mcd(u)< core(u)$, into a queue $Q$. All vertices in $Q$ need to update their core numbers based on Lemma~\ref{R3:mcd} (Lines~\ref{alg5_x1}-\ref{R3:alg5:2}). After that, the algorithm recursively probes each neighboring vertex $v$ of vertices in $Q$, and adds $v$ into the vertex set $V^*$ if $mcd(v) < core(v)$ (Lines~\ref{R3:alg5:3}-\ref{R3:alg5:4}). 
In the second step, we maintain the $K$-order $O'$ by adjusting the position of vertices in $V^*$, which is identified in Step 1, to reflect the edges deletion of $E^-$ (Lines~\ref{R3:alg5:5}-\ref{R3:alg5:6}). In details, for each $u\in V_i$, we update the $deg^+(.)$ of $u$ and its neighbours, remove $u$ from $\mathcal{O}'_t$, and insert $u$ to the end of $\mathcal{O}'_{t-1}$. In the final step, we use $V_R$ to record the vertices that may become the potential followers for the anchored vertices, i.e., these vertices' core number becomes $k-1$ in the new $K$-order $O'$ (Line~\ref{R3:alg5:7}). 

\subsection{The Incremental Algorithm}

\begin{algorithm}[t] \footnotesize 
	\caption{\textit{IncAVT}} \label{alg:IncAVT}
	\KwIn{$\mathcal{G}=\{G_t\}^T_1:$ an evolving graph, $l$: the allocated size of anchored vertex set, and $k$: degree constraint} 
	\KwOut{ $\mathcal{S}=\{S_t\}_1^T:$ the series of anchored vertex sets} 
	\BlankLine
	Build the $K$-order $O(G_1)$ of $G_1$; \textcolor{black}{/* using Algorithm~\ref{alg:k-core} */}  \label{R3_alg6_1}  \\
	Compute the anchored vertex set $S_1$ of $G_1$ with size $l$ using Algorithm~\ref{alg:basic1}; \\
	$\mathcal{S}:=\{S_1\}$; $t:= 2$; \\ \label{R3_alg6_2}
	\While{$t < T$}{ \label{while1}
	    $G'_t := G_{t-1}\oplus E^+$, 
	    $S_t\leftarrow S_{t-1}$; \\ \label{R3_alg6_3+}
        \textcolor{black}{/* maintain $K$-order by using Algorithm~\ref{alg:insert}, \ref{R3:alg:remove} */} \\
	    $(O', V_I)\leftarrow$ \textit{EdgeInsert}($G'_{t}, O(G_{t-1}), E^+, k$); \\  \label{R3_alg6_4}
	    $(O(G_t), V_R)\leftarrow$ \textit{EdgeRemove}($G'_{t}, O', E^-, k$);  \\  \label{R3_alg6_5}

	    \For {each $u\in S_{t-1}$}{ \label{R3_alg6_6}
	        compute $F_k(S_t, G_t)$, $F\leftarrow |F_k(S_t, G_t)|$; \\  \label{R3_alg6_6.1}
	        $F_{max} \leftarrow 0$, $u'\leftarrow u$; \\
	        \For{ each \textcolor{black}{ ~ ~ ~ ~ ~ ~ ~ ~ ~ ~ ~ ~ ~ /* Theorem~\ref{R3:theo1} */} 
	        $\{v|v\in \{V_I\cup V_R\cup nbr(V_I\cup V_R)\setminus C_k(G_t)\}\wedge \{\exists u\in nbr(v)\wedge core(u)=k-1\wedge v\preceq u\}\}$    }{  \label{R3_alg6_8} 
	             \If{$F_{max}< F_k(S_t\setminus u\cup v, G_t)$}{
                     $F_{max} \leftarrow F_k(S_t\setminus u\cup v, G_t)$, $u'\leftarrow v$; \\	              \label{R3_alg6_9}
	             }
	        }
	        \If{$F_{max}> F$}{  \label{R3_alg6_10}
	           remove $u$ from $S_t$, 
               add $u'$ to $S_t$; \\   \label{R3_alg6_7}	        
	        }
	    }
       $\mathcal{S}:= \mathcal{S}\cup S_t$; 
       $t\leftarrow t + 1$; \\  \label{while2}
     }
\Return{$\mathcal{S}$}	\label{R3_return} 
\end{algorithm}
\setlength{\textfloatsep}{0.2mm} 

Base on the above $K$-order maintenance strategies and the impacted vertex sets $V_I$ and $V_R$, we propose an efficient incremental algorithm, \textit{IncAVT}, for processing the \textit{AVT} query. Algorithm~\ref{alg:IncAVT} summarizes the major steps of \textit{IncAVT}. Given an evolving graph $\mathcal{G} = \{G_t\}^T_1$, the allocated size of selected anchored vertex set $l$, and a positive integer $k$, the \textit{IncAVT} algorithm returns a series of anchored vertex set $S= \{\mathcal{S}_t\}^T_1$ of $\mathcal{G}$ where each $S_t$ has size $l$. Initially, we build the $K$-order $O(G_1)$ of $G_1$ by using Algorithm~\ref{alg:k-core}, and then compute the anchored vertex set $S_1$ of $G_1$ by using Algorithm~\ref{alg:basic1} where $T$ is set as $1$ (Lines~\ref{R3_alg6_1}-\ref{R3_alg6_2}). 
The \textit{while} loop at lines~\ref{while1}-\ref{while2}, computes the anchored vertex set of each snapshot graph $G_t\in \mathcal{G}$. 
$E^+$ and $E^- $ represent the edges insertion and edges deletion between $G_{t-1}$ to $G_t$ respectively, and we initialize the anchored vertex set $S_{t}$ in $G_t$ as $S_{t-1}$ (Line~\ref{R3_alg6_3+}). The $K$-order is maintained by using Algorithm~\ref{alg:insert} while considering the edge insertion $E^+$ to $G_{t-1}$ and consequently, the vertex set $V_I$ is returned to record the vertices, which is impacted by inserting $E^+$ and has core number $k-1$ in the updated $K$-order (Line~\ref{R3_alg6_4}). Similarly, we use Algorithm~\ref{R3:alg:remove} to update the $K$-order while considering the edges deletion of $E^-$ and use $V_R$ to record the vertices which has core number k-1 and impacted by the edge deletion (Line~\ref{R3_alg6_5}). Next, an inner \textit{for} loop is to track the anchored vertex set of $G_{t}$ (Lines~\ref{R3_alg6_6}-\ref{R3_alg6_7}). More specifically, we first compute $S_t$'s followers set size $F$ (Line~\ref{R3_alg6_6.1}). Then, for each vertex $u$ in $S_{t-1}$, we only probe the vertices $v$ in vertex set $\{V_I\cup V_R \cup nbr(V_I\cup V_R)\setminus C_k(G_t)\}$ based on Theorem~\ref{R3:theo1} (Lines~\ref{R3_alg6_6}-\ref{R3_alg6_9}). If the number of followers of anchored vertex set $\{S_t\setminus u \cup v\}$ is \TT{bigger} than $F$, we then update $S_t$ by using $v$ to replacement $u$ (Lines~\ref{R3_alg6_10}-\ref{R3_alg6_7}). After the inner \textit{for} loop finished, we add the anchored vertex set $S_t$ of $G_t$ into $\mathcal{S}$ (Line~\ref{while2}). The \textit{IncAVT} algorithm finally returns the series of anchored vertex set $\mathcal{S}$ as the final result (Line~\ref{R3_return}). 


\section{Experimental Evaluation} \label{sec:exp}
In this section, we present the experimental evaluation of our proposed approaches for the AVT problem: \TT{the Greedy algorithm that is optimized by two strategies mentioned in Section~\ref{sec:alg}} (\textbf{Greedy}); and the incremental algorithm (\textbf{IncAVT}). 
\TT{The source codes of this work are available at \url{https://github.com/IncAVT/IncAVT}}.

\subsection{Experimental Setting}
\textbf{Algorithms.} 
To the best of our knowledge, no existing work investigates the \textit{Anchored Vertex Tracking} (AVT) problem. To further validate, we compare with two baselines adapted from the existing works: (i) \textbf{OLAK}, which is proposed in~\cite{DBLP:journals/pvldb/ZhangZZQL17} to find out the best anchored vertices at each snapshot graph, and (ii) \textbf{RCM}, which is the state-of-the-art anchored $k$-core algorithm proposed in~\cite{DBLP:conf/sdm/LaishramSEPS20}, for tracking the best anchored vertices selection at each snapshot graph.



\vspace{2pt}
\noindent
\textbf{Datasets.} We conduct the experiments using \TTao{six} publicly available datasets from the \textit{Stanford Network Analysis Project (SNAP)}\footnote{http://snap.stanford.edu/}: \textit{email-Enron}, \textit{Gnutella}, \textit{Deezer}, \TT{\textit{eu-core}, \textit{mathoverflow}, and \TTao{\textit{CollegeMsg}}}. 
The statistics of the datasets are shown in Table~\ref{tab:datasets}. As the orginal datasets \TT{(i.e., \textit{email-Enron}, \textit{Gnutella}, and \textit{Deezer})} do not contain temporal information, we thus generate 30 synthetic time evolving snapshots for each dataset by randomly inserting new edges and removing old edges.
More specifically, we use it as the first snapshot $T_1$. Then, we randomly remove $100 - 250$ edges from $T_1$, denoted as $T'_1$ and randomly add $100-250$ new edges into $T'_1$, denoted as $T_2$. By repeating the similar operation, we generate $30$ snapshots for each dataset.
Moreover, we further conduct our experiments using \TTao{two} real-world temporal network datasets from SNAP: \textit{en-core}, \textit{mathoverflow}, and \TTao{\textit{CollegeMsg}}. Specifically, we have averagely divided these two datasets into $T$ graph snapshots (\textit{e.g., $G_t = (V, E_t)$, $t\in [0,T]$)}, where $V$ is the vertex and $E_t$ is the edges appearing in the time period of $t$ in each dataset. Besides, the edge insertion set $E^+$ of $G_t$ contains edges newerly appears in $G_t$ but does not exist in $G_{t-1}$; Similarly, the edge deletion set $E^-$ of $G_t$ is the edges existed in $G_{t-1}$ but disappear in $G_t$. Note that an edge will be disppear if it keeps being inactive in a period of time (\textit{i.e., a time window $W = 365$ days in \textit{mathoverflow} dataset}). 

\begin{table}[t]
	\caption{\TT{Dataset Statistics} \label{tab:datasets}}
	\vspace{-6mm}
	\begin{center}
		\scalebox{0.9}{
			\begin{tabular}{|c|c|p{1.1cm}|c|c|c|c|} \hline
				{\bf Dataset}  & Nodes & \TT{(Temporal)} Edges  & $d_{avg}$ & Days & Type      \\\hline
				\hline
				email-Enron    & 36,692    & 183,831    & 10.02 & -  & Communication    \\\hline
				Gnutella       & 62,586    & 147,878    & 4.73 & -   & P2P Network    \\\hline 
				Deezer         & 41,773    & 125,826    & 6.02 & -   & Social Network   \\\hline		
			\TT{eu-core}        & \TT{986}       & \TT{332,334}    & \TT{25.28} & \TT{803}    & \TT{Email}  \\\hline	
			\TT{mathoverflow}  & \TT{13,840}    & \TT{195,330}    & \TT{5.86} & \TT{2,350}    & \TT{Question\&Answer}   \\\hline
			\TTao{CollegeMsg}  & \TTao{1,899}    & \TTao{59,835}    & \TTao{10.69} & \TTao{193}    & \TTao{Social Network}   \\\hline
			\end{tabular}
		}
	\end{center}
	\vspace{-2mm}
\end{table}

\begin{table}[t]
	\caption{Parameters and Their Values \label{R3_tab:parameter1}}
	\vspace{-3mm}
	\begin{center}
		\scalebox{0.9}{
			\begin{tabular}{|c|p{4.0cm}<{\centering}|c|} \hline
				{\bf Parameter}  & Values                         & Default \\\hline 
				\hline
				$l$              & $[5, 10, 15, 20]$                        & $10$       \\\hline
				$k$              & $[2,3,4,5]$ or $[5,10,15,20]$    & $3$ or $10$   \\\hline 
				$T$              & $[0 - 30]$                       & 30        \\\hline
			\end{tabular}
		}
	\end{center}
	\vspace{-2mm}
\end{table}

\vspace{2pt}
\noindent
\textbf{Parameter Configuration.} Table~\ref{R3_tab:parameter1} presents the parameter settings.
We consider three parameters in our experiments: core number $k$, anchored vertex size $l$, and the number of snapshots $T$. In each experiment, if one parameter varies, we use the default values for the other parameters. Besides, we use the sequential version of the RCM algorithm in the following discussion and results. All the programs are implemented in C++ and compiled with GCC
on Linux. The experiments are executed on the same computing server with 2.60GHz Intel Xeon CPU and 96GB RAM.

\subsection{Efficiency Evaluation}
In this section, we  study the efficiency of the approaches for the AVT problem regarding running time under different parameter settings. 

\subsubsection{Varying Core Number $k$}
We compare the performance of different approaches by varying $k$. Due to the various average degree of \TTao{six} datasets, we set different $k$ for them.
Figure~\ref{R3:exp:vary_k1} - \TTao{\ref{R3:exp:vary_k6}} show the running time of \textit{OLAK}, \textit{Greedy}, \textit{IncAVT}, and \textit{RCM}, on the \TTao{six} datasets. 
From the results, we can see that \textit{Greedy} and \textit{RCM} perform faster than \textit{OLAK}, and \textit{IncAVT} performs one to two orders of magnitude faster than the other three approaches in \TT{\textit{email-Enron, Gnutella, and Deezer}}. \TT{Besides, our proposed Greedy method performs the best in \textit{eu-core}, \textit{mathoverflow}}\TTao{, and \textit{CollegeMsg}. }
As expected, we do not observe any noticeable trend from all three approaches when $k$ is varied. This is because, in some networks, the increase of the core number may not induce the \TT{increase} of the size of $k$-core subgraph and the number of candidate anchored vertices needing to probe. 

\begin{figure}[ht]
	\centering
	\subfigure[email-Enron]{\label{R3:exp:vary_k1}
		\includegraphics[width=0.3\linewidth]{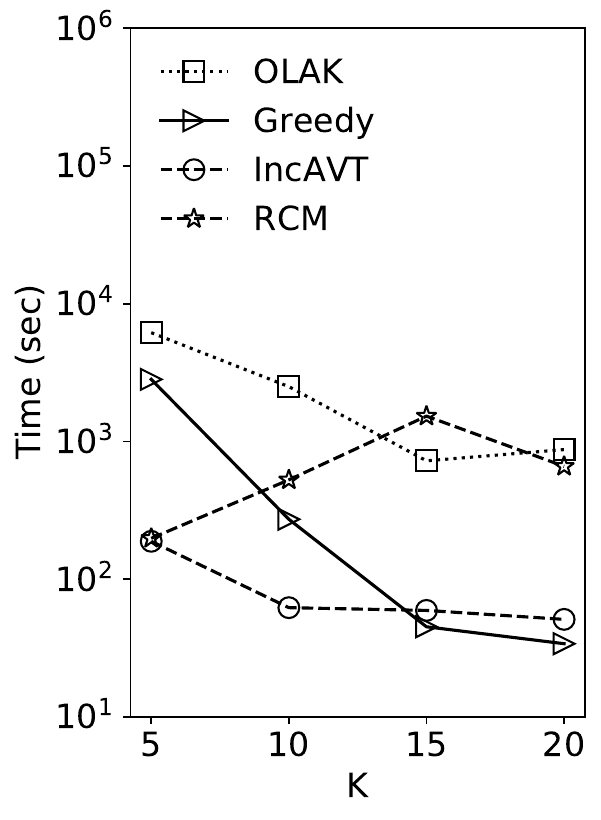}
	}
	\subfigure[Gnutella]{\label{R3:exp:vary_k2}		
		\includegraphics[width=0.3\linewidth]{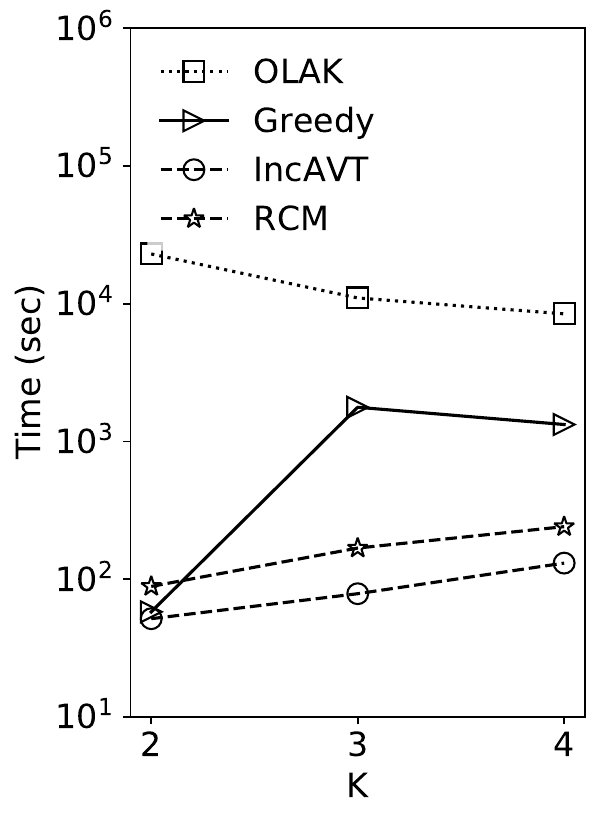}
	}
	\subfigure[Deezer]{\label{R3:exp:vary_k3}		
	\includegraphics[width=0.3\linewidth]{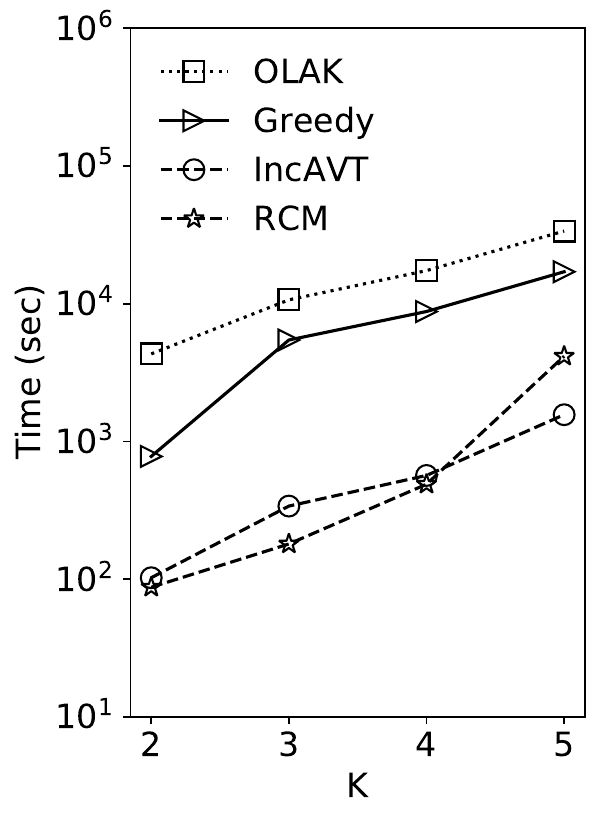}
	}
	\subfigure[eu-core]{\label{R3:exp:vary_k4}		
	\includegraphics[width=0.3\linewidth]{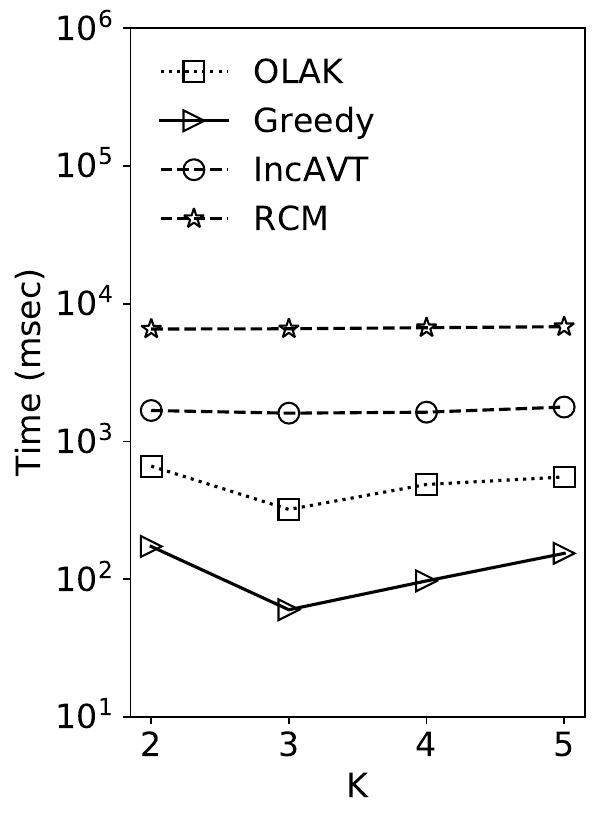}
	}
	\subfigure[mathoverflow]{\label{R3:exp:vary_k5}		
	\includegraphics[width=0.3\linewidth]{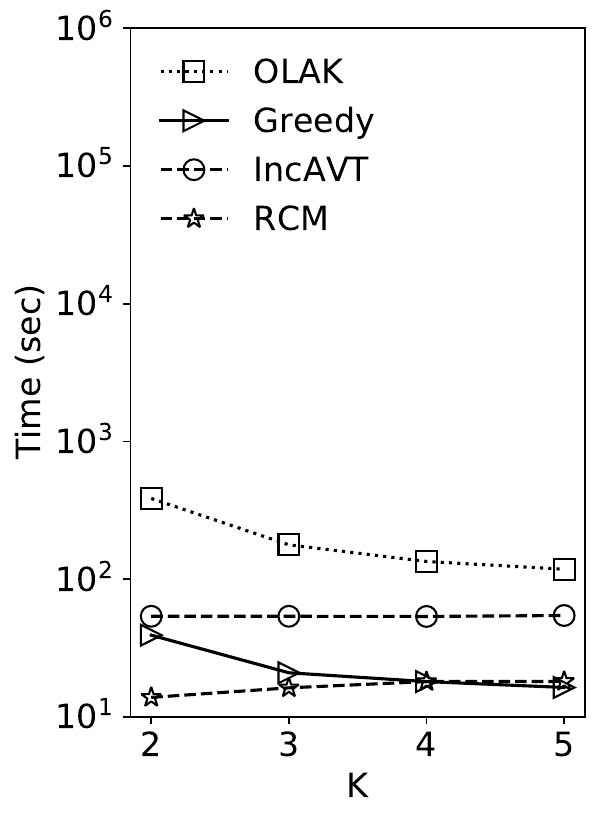}
	}
	\subfigure[\TTao{CollegeMsg}]{\label{R3:exp:vary_k6}		
	\includegraphics[width=0.3\linewidth]{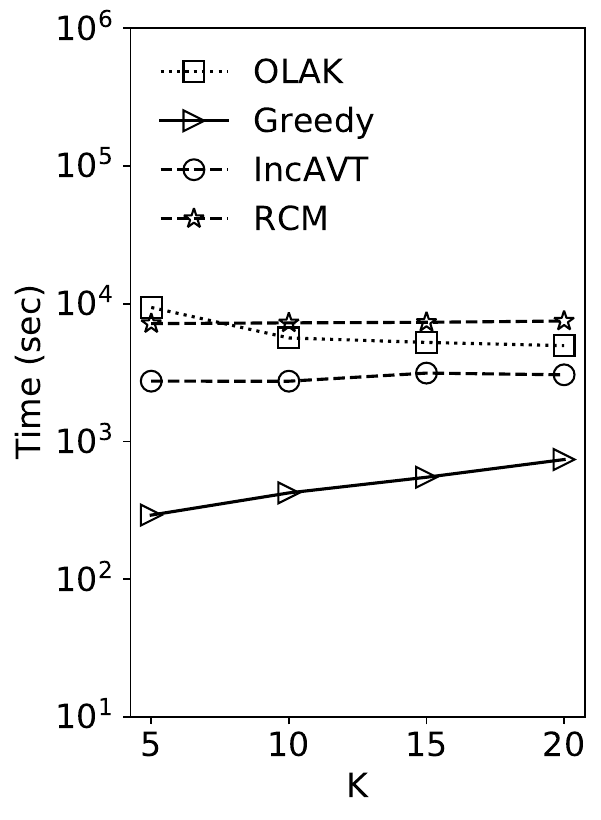}
	}
	\vspace{-2mm}
	\caption{\TT{Time cost of algorithms with varying $k$}}
	\label{fig:running_time_k}
\end{figure}

Since the performance of \textit{Greedy}, \textit{OLAK}, and \textit{IncAVT} are highly influenced by the number of visited candidate anchored vertices in algorithm execution, we also investigate the number of candidate anchored vertices that need to be probed for these approaches in different datasets. 
Figure \ref{R3_exp:anchor_vary_k1} - \TTao{\ref{R3_exp:anchor_vary_k6}} show the number of visited candidate anchored vertices for the three approaches when $k$ is varied. 
We notice that \textit{OLAK} visits more number of candidate anchored vertices than the other two approaches, and \textit{IncAVT} shows the minimum number of visited candidate anchored vertices.  

\begin{figure}[!tb]
	\centering
	\subfigure[email-Enron]{	\label{R3_exp:anchor_vary_k1}
		\includegraphics[width=0.30\linewidth]{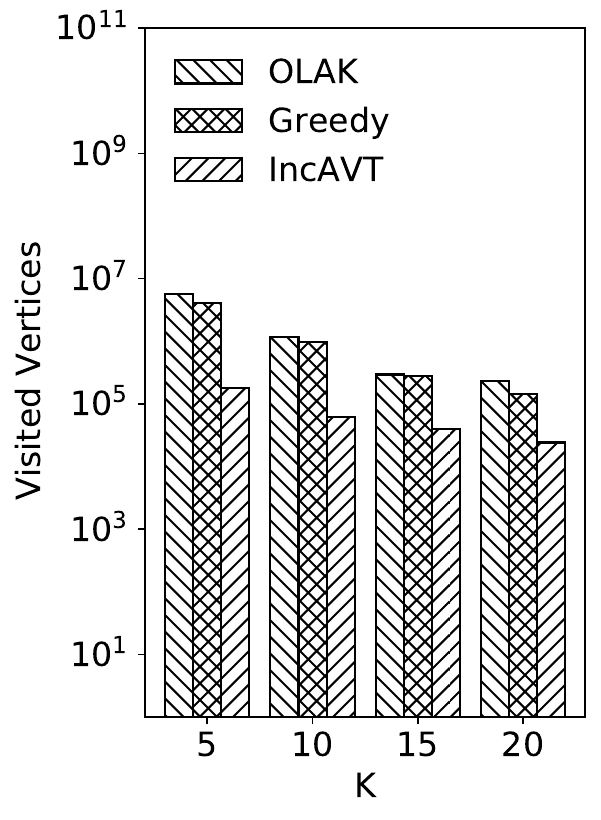}
	}
	\subfigure[Gnutella]{	\label{R3_exp:anchor_vary_k2}		
		\includegraphics[width=0.30\linewidth]{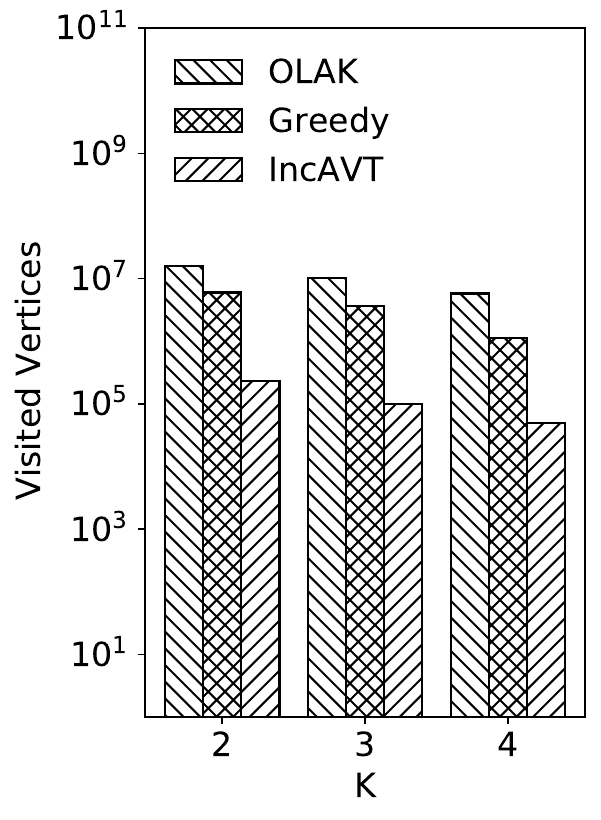}
	}
	\subfigure[Deezer]{	\label{R3_exp:anchor_vary_k3}		
		\includegraphics[width=0.30\linewidth]{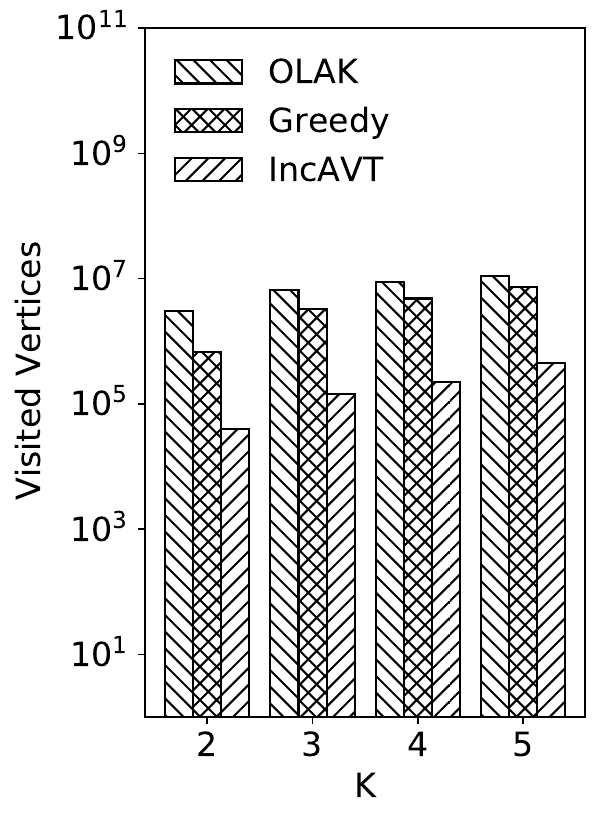}
	}
	\subfigure[eu-core]{	\label{R3_exp:anchor_vary_k4}		
		\includegraphics[width=0.30\linewidth]{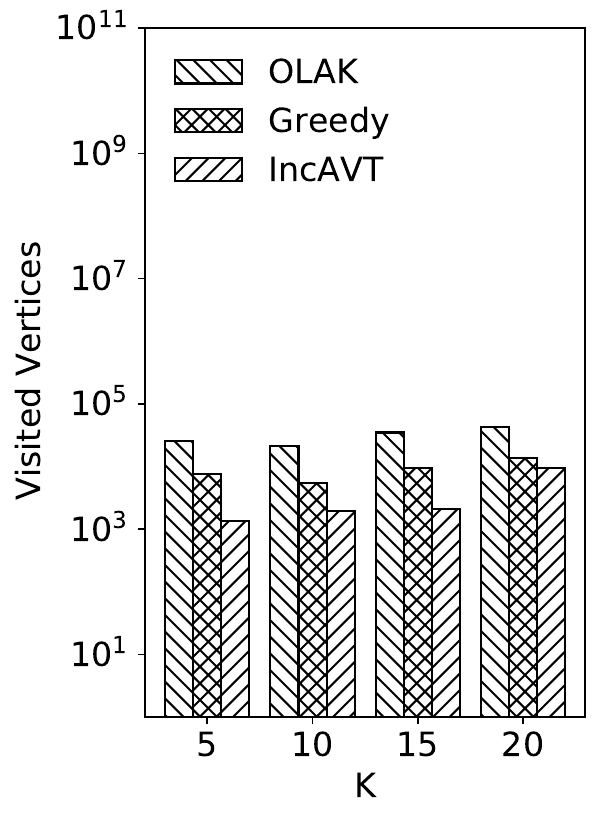}
	}
	\subfigure[mathoverflow]{	\label{R3_exp:anchor_vary_k5}		
		\includegraphics[width=0.30\linewidth]{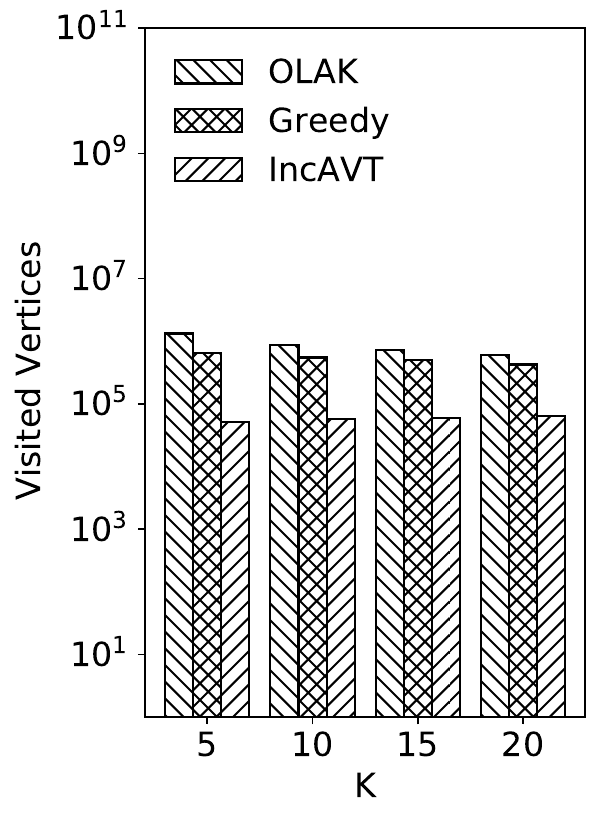}
	}
	\subfigure[\TTao{CollegeMsg}]{	\label{R3_exp:anchor_vary_k6}		
		\includegraphics[width=0.30\linewidth]{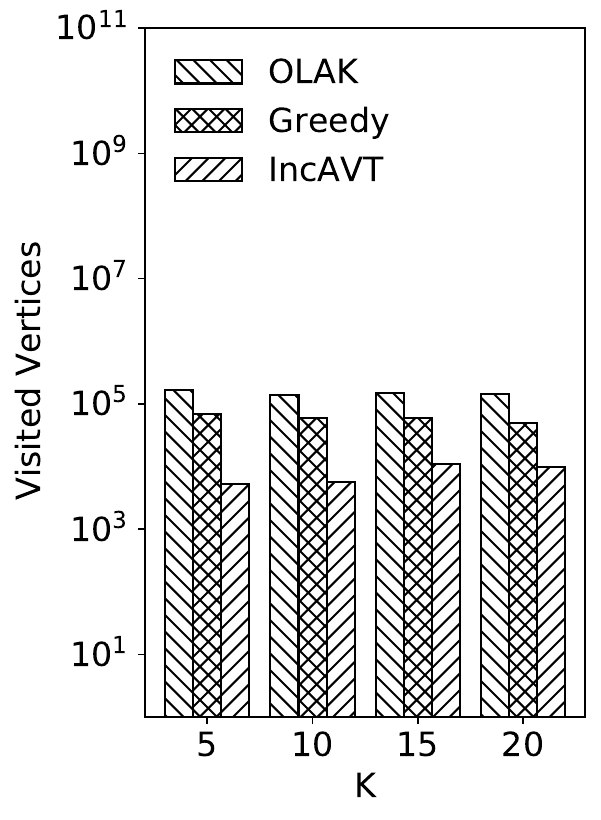}
	}
		\vspace{-2mm}
	\caption{\TT{Number of candidate anchored vertices with varying $k$}}
	
	\label{fig:visited_candidates_k}
\end{figure}

\subsubsection{Varying Snapshot Size $T$}
We also test our proposed algorithms by varying $T$ from 2 to 30.
Specifically, Figure~\ref{R3:exp:vary_t1} - \ref{R3:exp:vary_t3} present 
the running time with varied values of $T$ \TT{in \textit{email-Enron, Gnutella, and Deezer}}. 
The results show similar findings that \textit{IncAVT} outperforms \textit{OLAK}, \textit{Greedy}, and \textit{RCM} significantly in efficiency as it utilizes the smoothness of the network structure in evolving network to reduce the visited candidate anchored vertices. 
Meanwhile, the speed of running time increasing in \textit{IncAVT} is much slower than the other three algorithms in each snapshot when $T$ increases. In other words, the performance advantage of \textit{IncAVT} will enhance with the increase of the network snapshot size. \TT{The above experimental results verify the excellent performance of our \textit{IncAVT} when the network is smoothly evolving, which is claimed in the contributions part of Section~\ref{sec:intro} in this paper.}  
Figure~\ref{R3:exp:vary_t4} - \TTao{\ref{R3:exp:vary_t6}} show the running time of these approaches on \TTao{three} real-world temporal datasets \textit{eu-core}, \textit{mathoverflow}\TTao{, and \textit{CollegeMsg}} when $T$ is varied. We observe that our optimized \textit{Greedy} method always performs better than \textit{OLAK} and \textit{RCM} for all varied $T$ values in \textit{eu-core} and \textit{mathoverflow}. As expected, in \textit{eu-core}, when $T \leq 20$, the performance of \textit{IncAVT} is significantly better than the other three methods; Besides, the running time of \TT{IncAVT} significantly increases when $T = 21$, and then increased slowly with the increases of $T$. This is because the efficiency of $K$-order maintenance will downgrade when the percentage of updated edges is high (\textit{i.e., 17\% percentage of edges updated at snapshot $T = 21$ in \textit{eu-core})}. In fact, the above phenomenon is the inherent character of the core maintenance technical strategy \textit{(e.g., Zhang et al.~\cite{DBLP:conf/icde/ZhangYZQ17} reported that their core maintenance related method decreased above five times when the percentage of updated edges increasing from $1\%$ to $5\%$)}. In addition, Figure~\ref{R3:exp:vary_t5} - \TTao{Figure~\ref{R3:exp:vary_t6}} show that even the performance of our \textit{IncAVT} method decreases at $T= 16$ in \textit{mathoverflow} \TTao{and $T = 22$ in \textit{CollegeMsg}}, when many edges are updated in \TTao{these two periods}, \textit{IncAVT} still performs better than \textit{OLAK} for all values of $T$. 


\begin{figure}[t!]
	\centering
	\subfigure[email-Enron]{	\label{R3:exp:vary_t1}
		\includegraphics[width=0.3\linewidth]{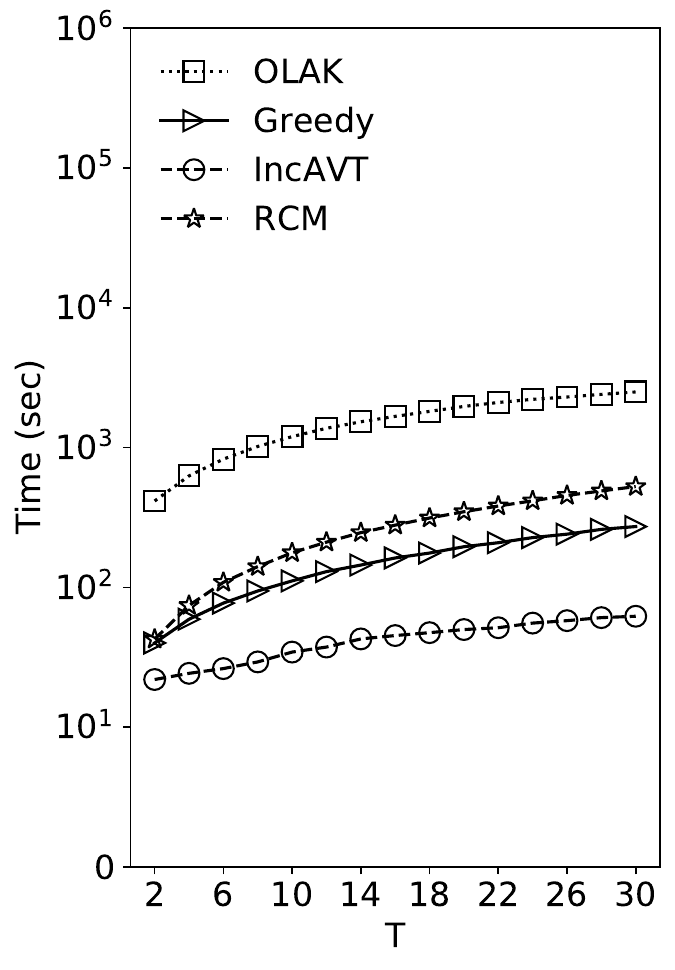}
	}
	\subfigure[Gnutella]{	\label{R3:exp:vary_t2}		
		\includegraphics[width=0.3\linewidth]{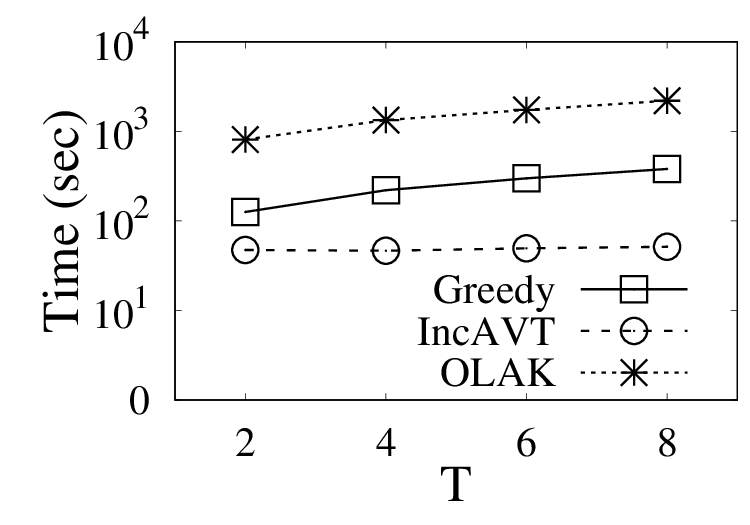}
	}
	\subfigure[Deezer]{	\label{R3:exp:vary_t3}
		\includegraphics[width=0.3\linewidth]{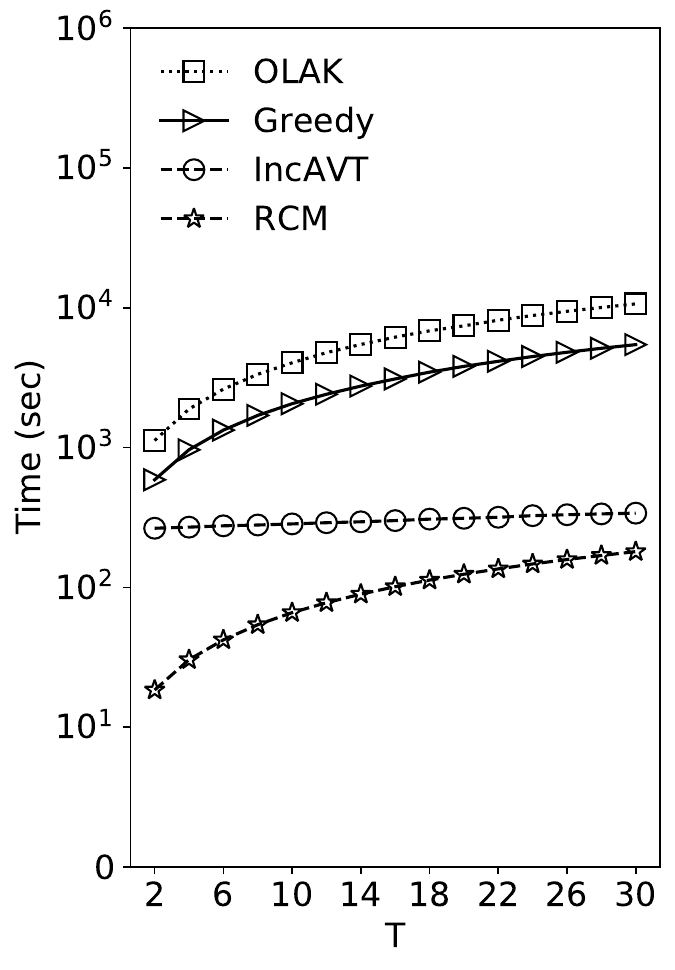}
	}
	\subfigure[eu-core]{	\label{R3:exp:vary_t4}		
		\includegraphics[width=0.3\linewidth]{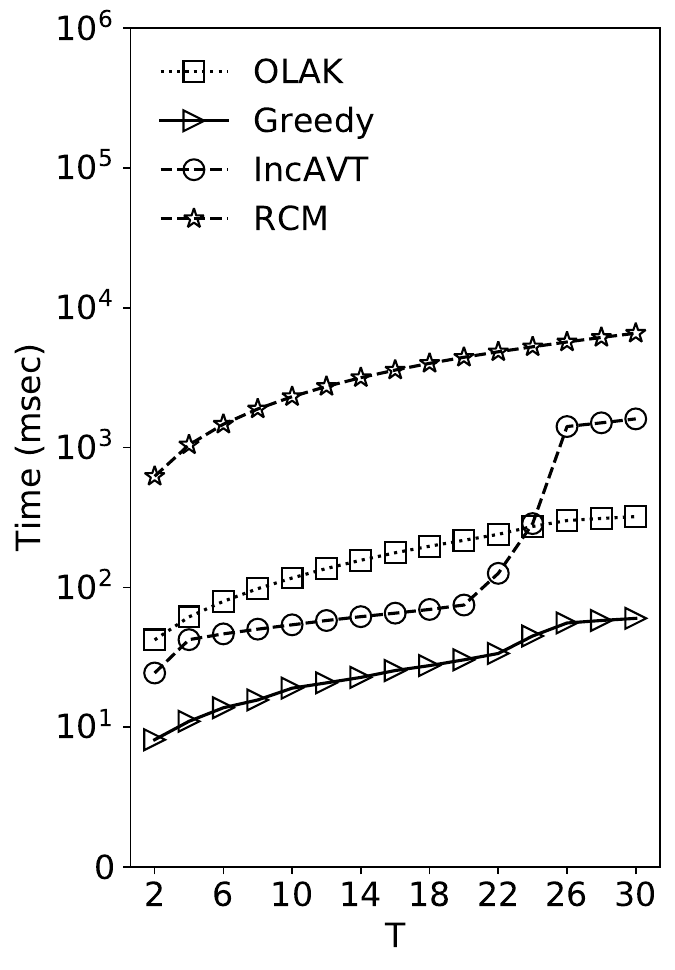}
	}
	\subfigure[mathoverflow]{	\label{R3:exp:vary_t5}
		\includegraphics[width=0.3\linewidth]{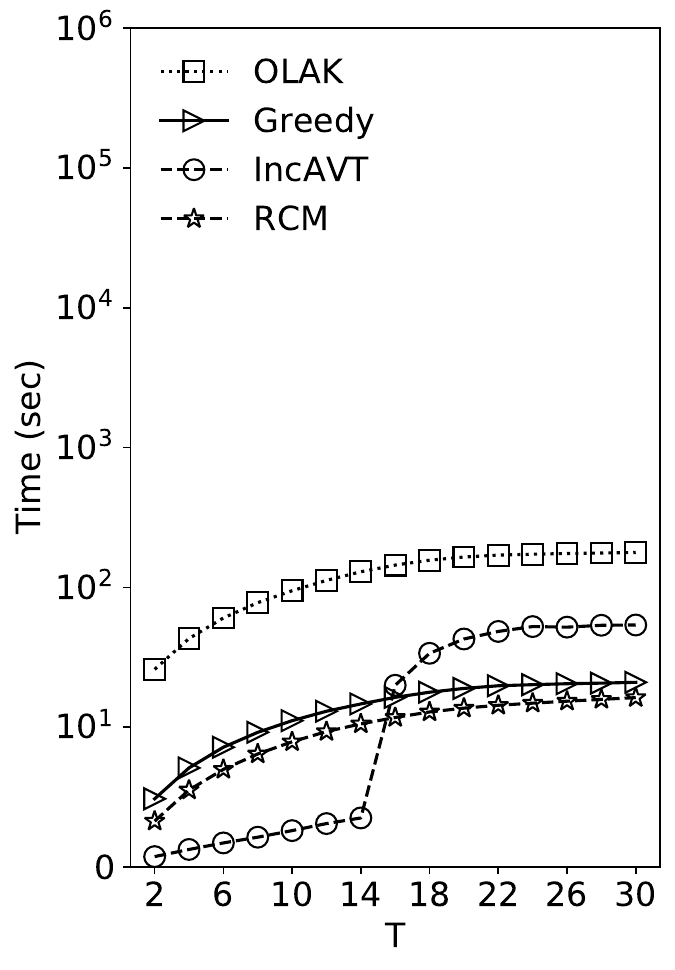}
	}
	\subfigure[\TTao{CollegeMsg}]{	\label{R3:exp:vary_t6}
		\includegraphics[width=0.3\linewidth]{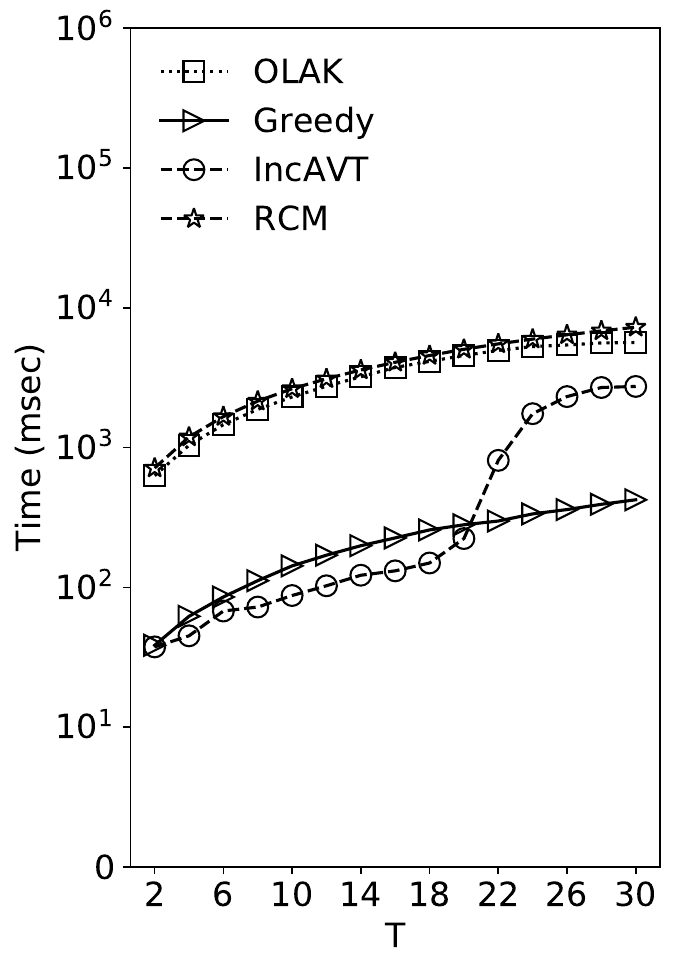}
	}
		\vspace{-2mm}
	\caption{\TT{Time cost of algorithms with varying $T$}}
	\label{fig:running_time_t}
\end{figure}

Figure~\ref{R3_exp:anchor_vary_t1} - \TTao{\ref{R3_exp:anchor_vary_t6}} report our further evaluation on the number of visited candidate anchored vertices when $T$ is varied.  
As expected, \textit{IncAVT} has the minimum number of visited candidate anchored vertices than the other two approaches. What is more, the number of visited candidate anchored vertices by \textit{IncAVT} in each snapshot is steady than \textit{Greedy} and \textit{OLAK}.

\begin{figure}[t!]
	\centering
		\subfigure[email-Enron]{	\label{R3_exp:anchor_vary_t1}		
		\includegraphics[width=0.3\linewidth]{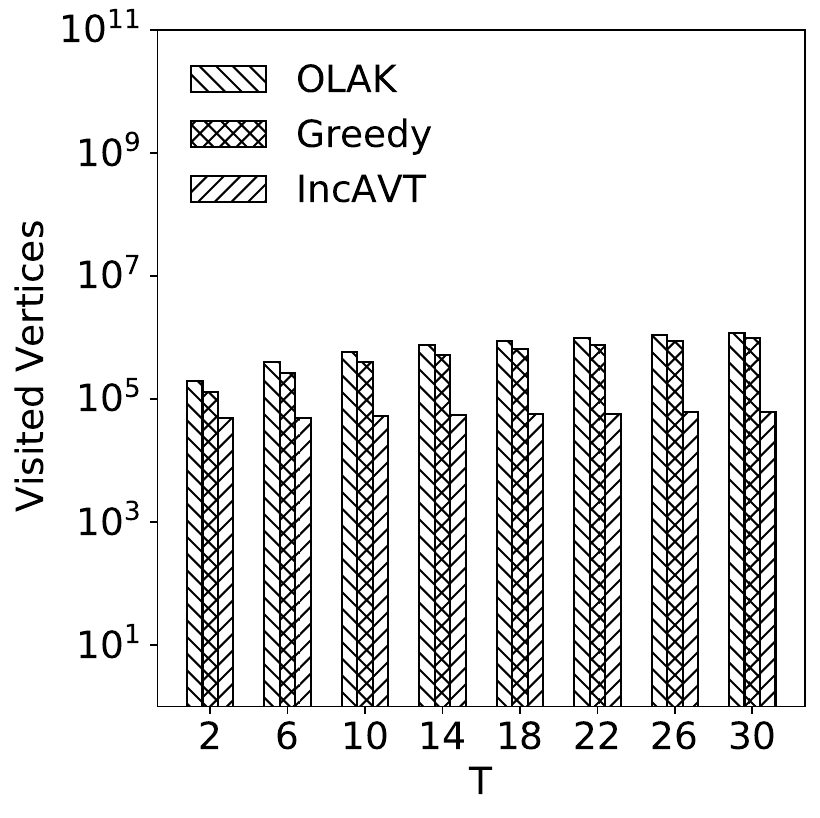}
	}
	\subfigure[Gnutella]{	\label{R3_exp:anchor_vary_t2}
		\includegraphics[width=0.3\linewidth]{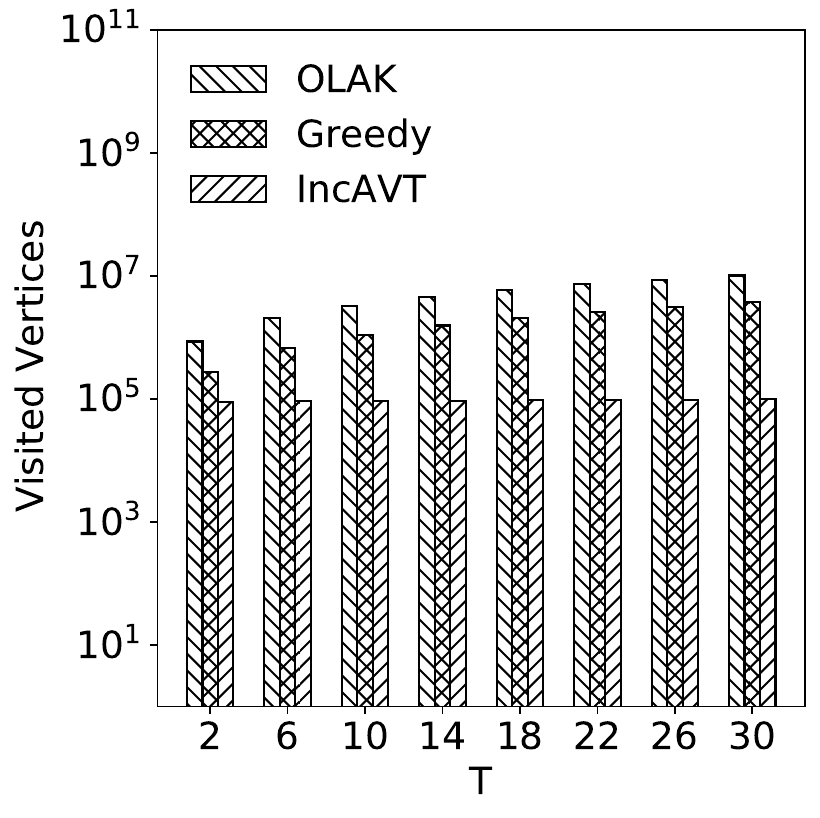}
	}
	\subfigure[Deezer]{	\label{R3_exp:anchor_vary_t3}
	\includegraphics[width=0.3\linewidth]{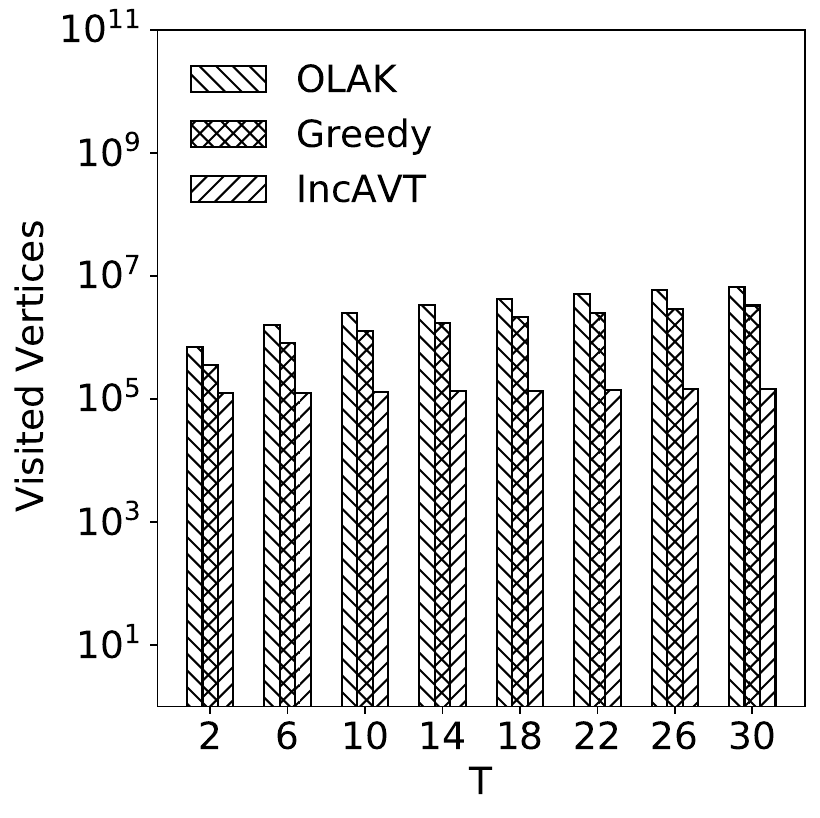}
	}
	\subfigure[eu-core]{	\label{R3_exp:anchor_vary_t4}		
		\includegraphics[width=0.3\linewidth]{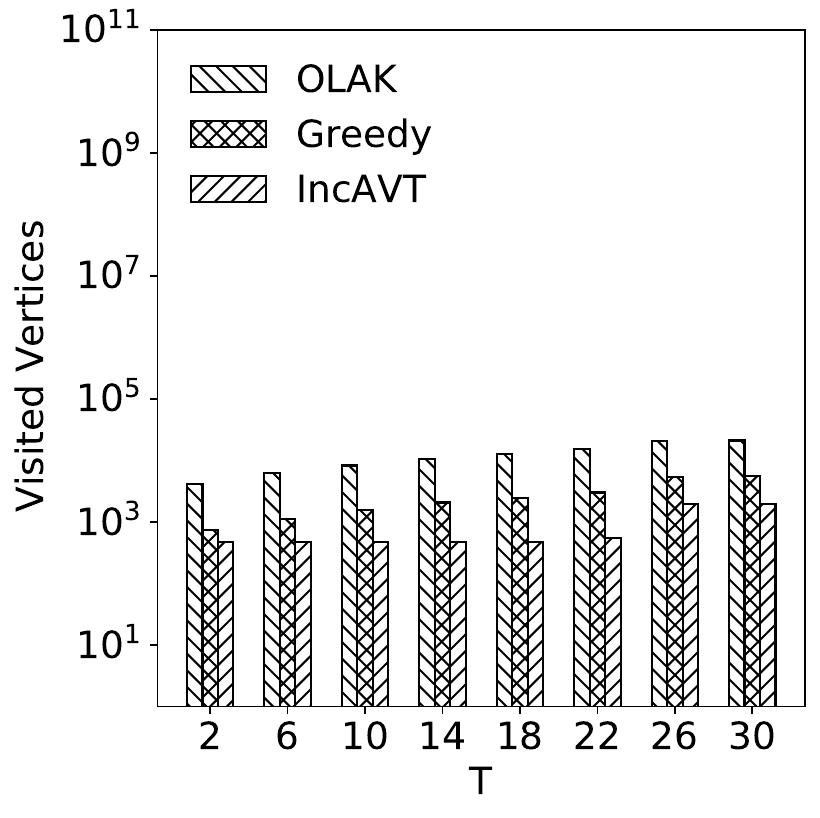}
	}
	\subfigure[mathoverflow]{	\label{R3_exp:anchor_vary_t5}		
		\includegraphics[width=0.3\linewidth]{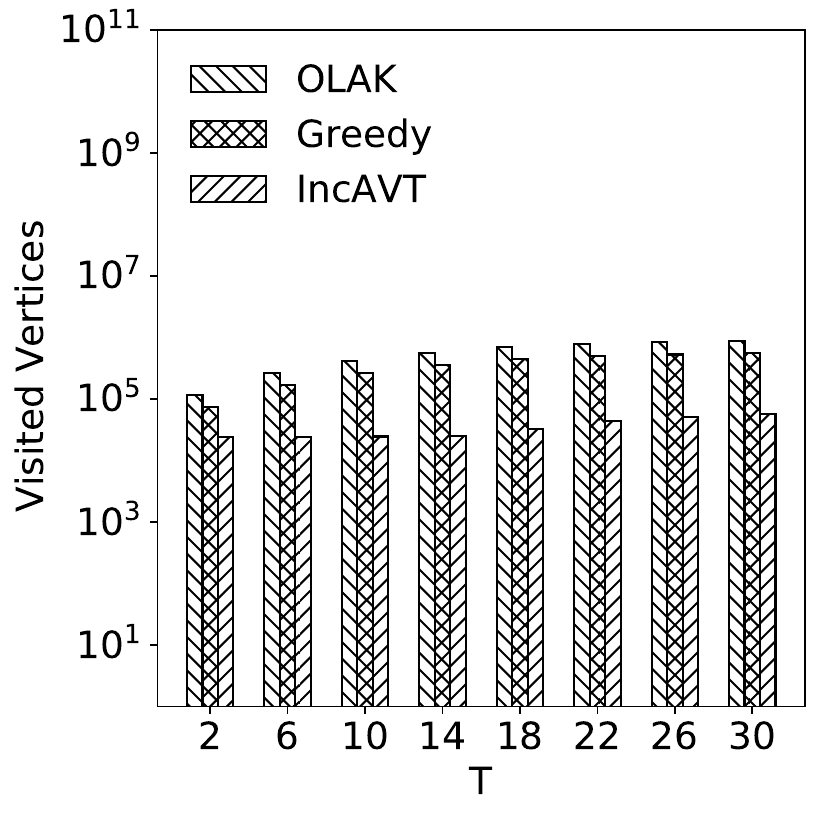}
	}
	\subfigure[\TTao{CollegeMsg}]{	\label{R3_exp:anchor_vary_t6}		
		\includegraphics[width=0.3\linewidth]{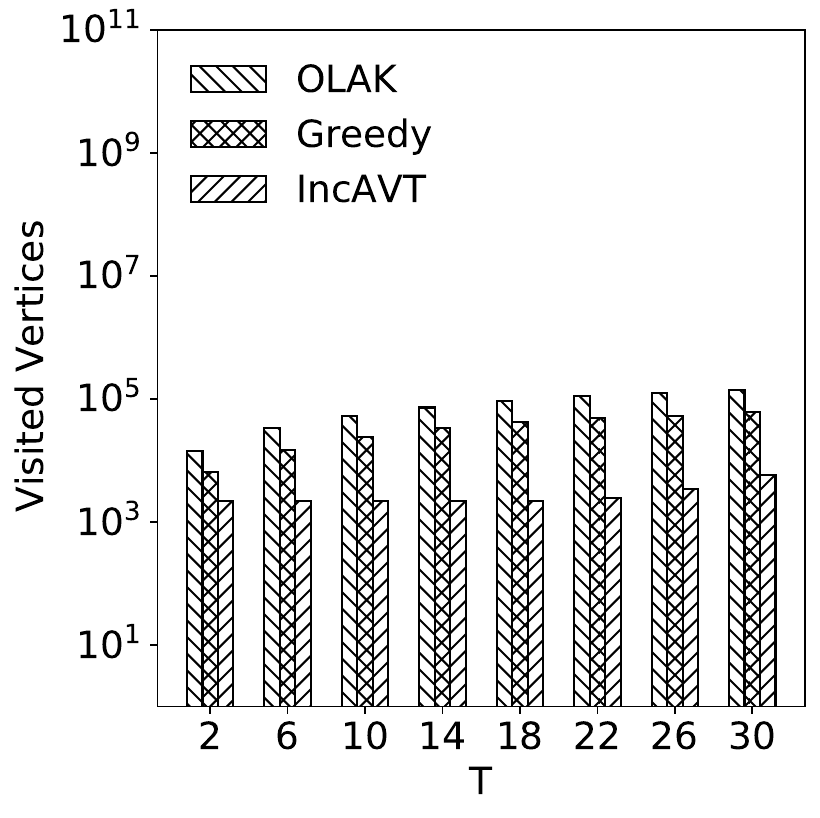}
	}
		\vspace{-2mm}
	\caption{\TT{Number of candidate anchored vertices with varying $T$}}
	\label{fig:visited_candidates_t}
\end{figure}

\subsubsection{Varying Anchored Vertex Set Size $l$}
Figure \ref{R3:exp:vary_l1} - \TTao{\ref{R3:exp:vary_l6}} show the average running time of the approaches by varying $l$ from 5 to 20. 
As we can see, \textit{IncAVT} is significantly efficient than \textit{Greedy} and \textit{OLAK} in \TT{\textit{email-Enron, Gnutella, Deezer, eu-core}, \textit{mathoverflow}}, \TTao{and \textit{CollegeMsg}}. Specifically, \textit{IncAVT} can reduce the running time by around 36 times and 230 times compared with \textit{Greedy} and \textit{OLAK} respectively under different $l$ settings on the \textit{Gnutella} dataset.
The improvements are built on the facts that \textit{IncAVT} visits less number of candidate anchored vertices than \textit{Greedy} and \textit{OLAK}.  Besides, \textit{IncAVT} performs far well than \textit{RCM} in \textit{Enron} and \textit{Gnutella}. Meanwhile, the running time of \textit{IncAVT} is slightly higher than \textit{RCM} in \textit{Deezer}. From the result, we notice that the performance of above approaches are also influenced by the type of networks.

\begin{figure}[t!]
	\centering
	\subfigure[email-Enron]{	\label{R3:exp:vary_l1}
		\includegraphics[width=0.3\linewidth]{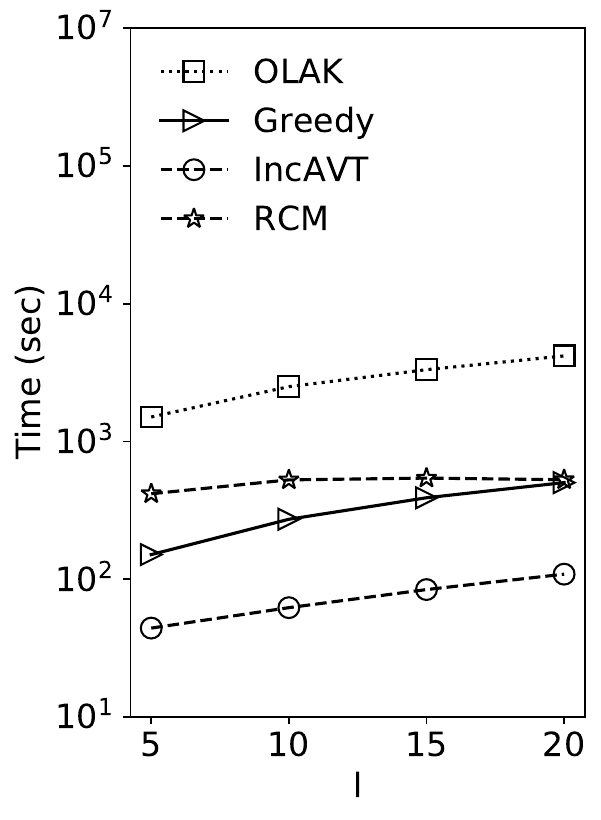}
	}
	\subfigure[Gnutella]{	\label{R3:exp:vary_l2}		
		\includegraphics[width=0.3\linewidth]{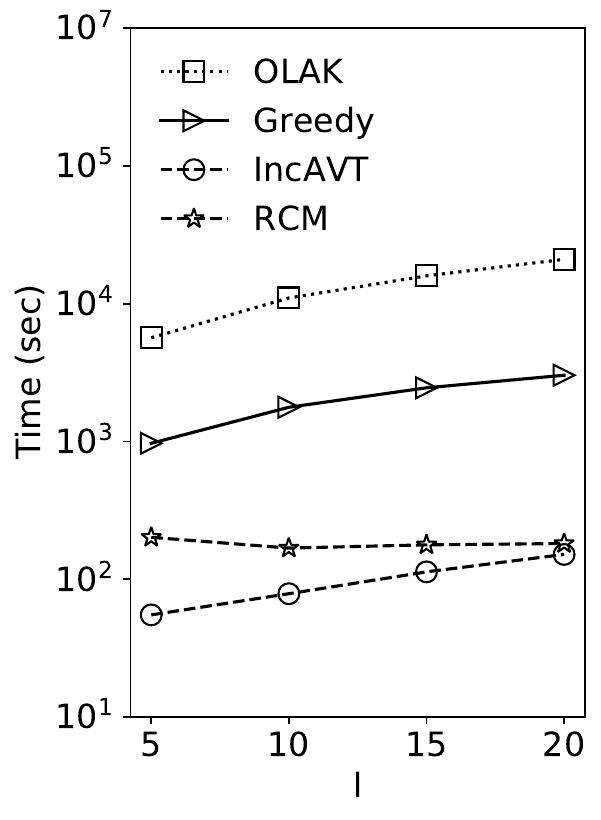}
	}
	\subfigure[Deezer]{	\label{R3:exp:vary_l3}		
		\includegraphics[width=0.3\linewidth]{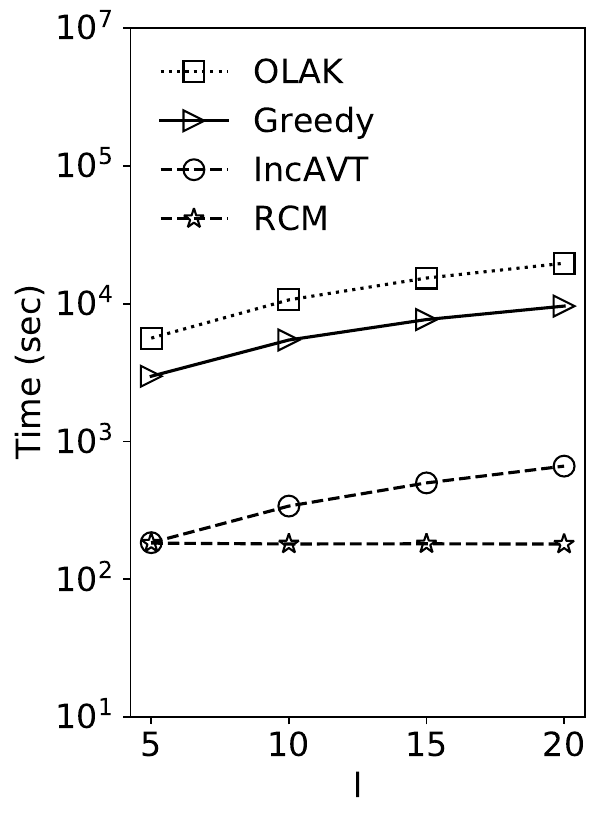}
	}
	\subfigure[eu-core]{	\label{R3:exp:vary_l4}		
		\includegraphics[width=0.3\linewidth]{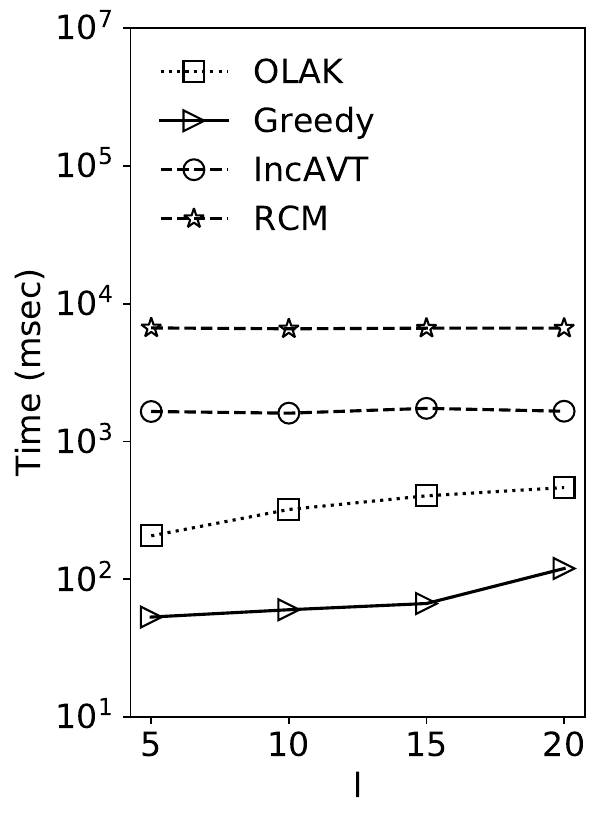}
	}
	\subfigure[mathoverflow]{	\label{R3:exp:vary_l5}		
		\includegraphics[width=0.3\linewidth]{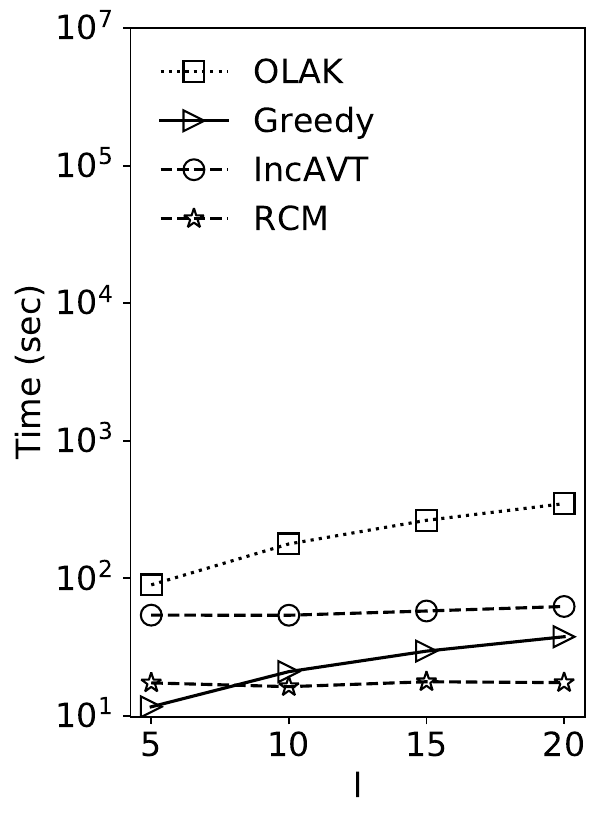}
	}
	\subfigure[\TTao{CollegeMsg}]{	\label{R3:exp:vary_l6}		
		\includegraphics[width=0.3\linewidth]{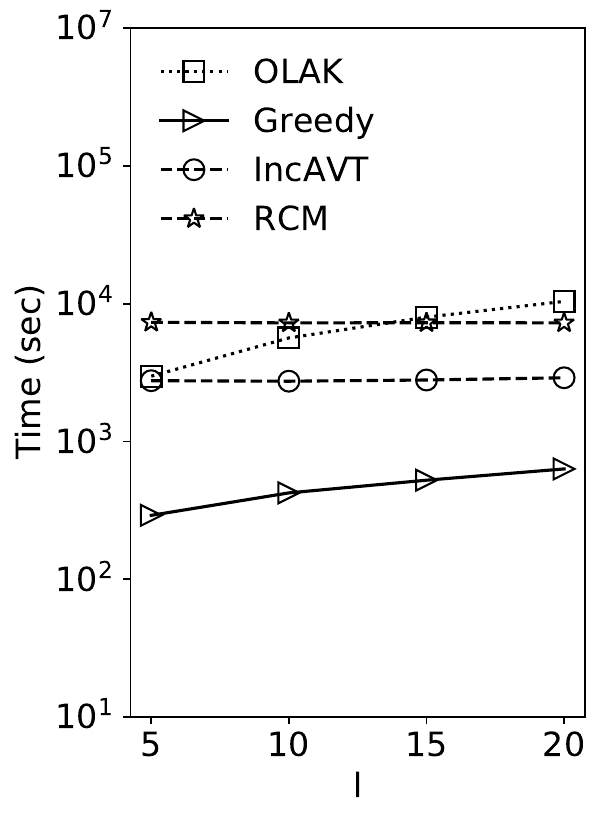}
	}
		\vspace{-2mm}
	\caption{\TT{Time cost of algorithms with varying $l$}}
	\label{fig:runing_time_l}
\end{figure}

Figure~\ref{R3:exp:anchor_vary_l1} - \TTao{\ref{R3:exp:anchor_vary_l6}} show the total number of visited anchored vertices.   
We can see that \textit{IncAVT} visits much less anchored vertices than the other two methods even though it shows a slightly increased number of visited vertices as $l$ increases.
The visited candidate anchored vertices in \textit{OLAK} is around 2.8 times more than \textit{Greedy}, and 102 times more than \textit{IncAVT} on the Gnutella dataset. 
The total number of visited candidate anchored vertex set in \textit{IncAVT} is minimum during the anchored vertex tracking process across all the datasets.

\begin{figure}[t!]
	\centering
	\subfigure[email-Enron]{	\label{R3:exp:anchor_vary_l1}		
		\includegraphics[width=0.3\linewidth]{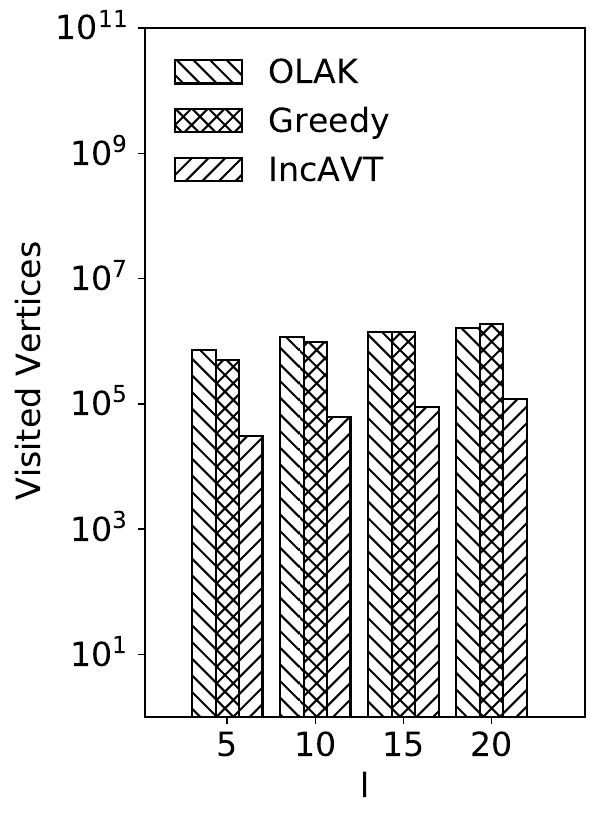}
	}
		\subfigure[Gnutella]{	\label{R3:exp:anchor_vary_l2}
		\includegraphics[width=0.3\linewidth]{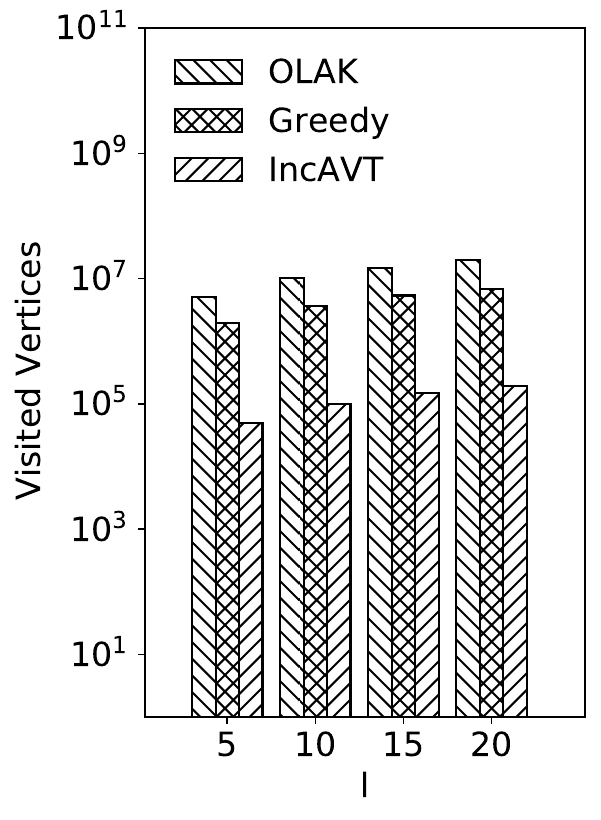}
	}
	\subfigure[Deezer]{	\label{R3:exp:anchor_vary_l3}
		\includegraphics[width=0.3\linewidth]{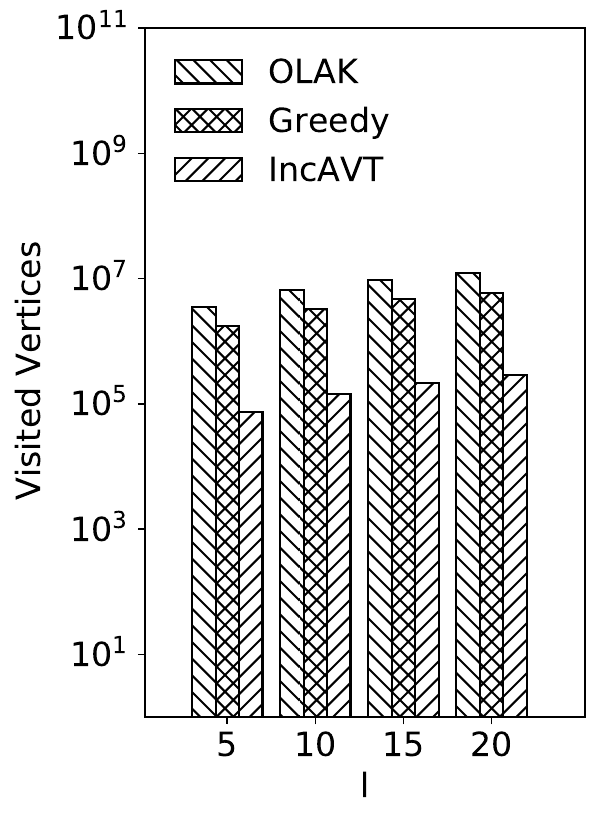}
	}
	\subfigure[eu-core]{	\label{R3:exp:anchor_vary_l4}		
		\includegraphics[width=0.3\linewidth]{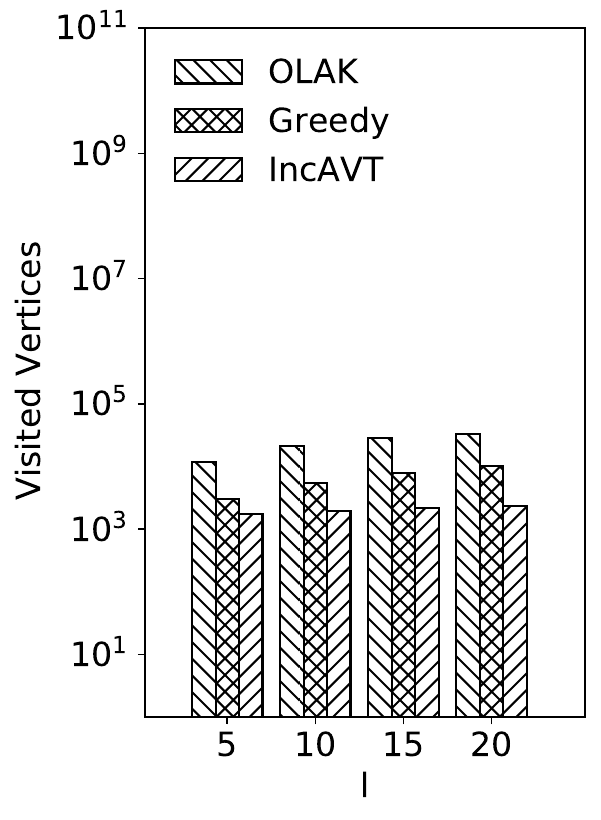}
	}
	\subfigure[mathoverflow]{	\label{R3:exp:anchor_vary_l5}		
		\includegraphics[width=0.3\linewidth]{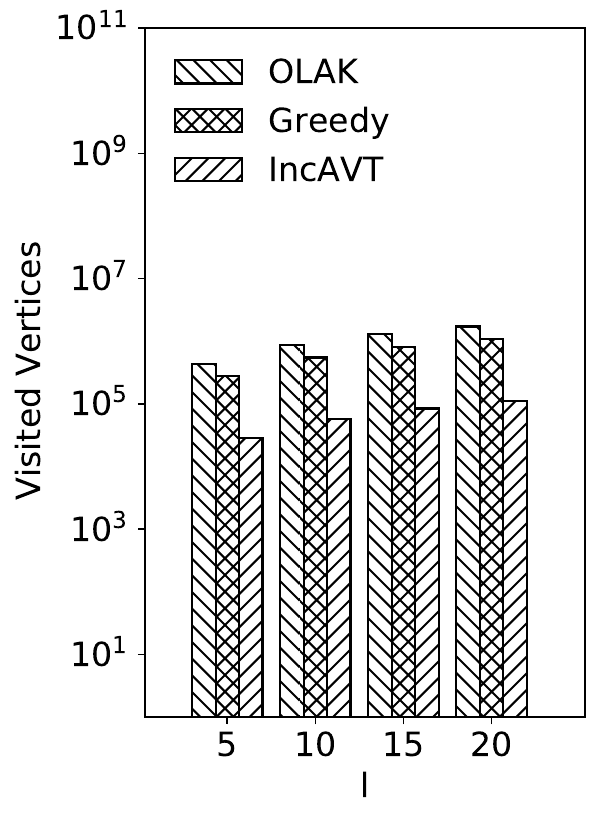}
	}
		\subfigure[\TTao{CollegeMsg}]{	\label{R3:exp:anchor_vary_l6}		
		\includegraphics[width=0.3\linewidth]{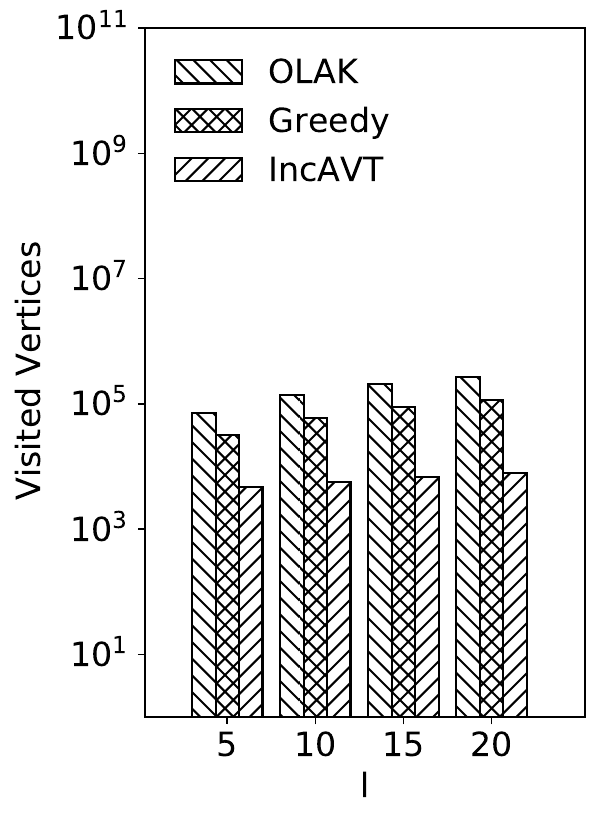}
	}
		\vspace{-2mm}
	\caption{\TT{Number of candidate anchored vertices with varying $l$}}
	\label{fig:visited_candidates_l}
	\vspace{-3mm}
\end{figure}

\subsection{Effectiveness Evaluation}
\begin{figure}[t!]
	\centering
	\subfigure[email-Enron]{	\label{R3:exp_follows:vary_t1}
		\includegraphics[width=0.3\linewidth]{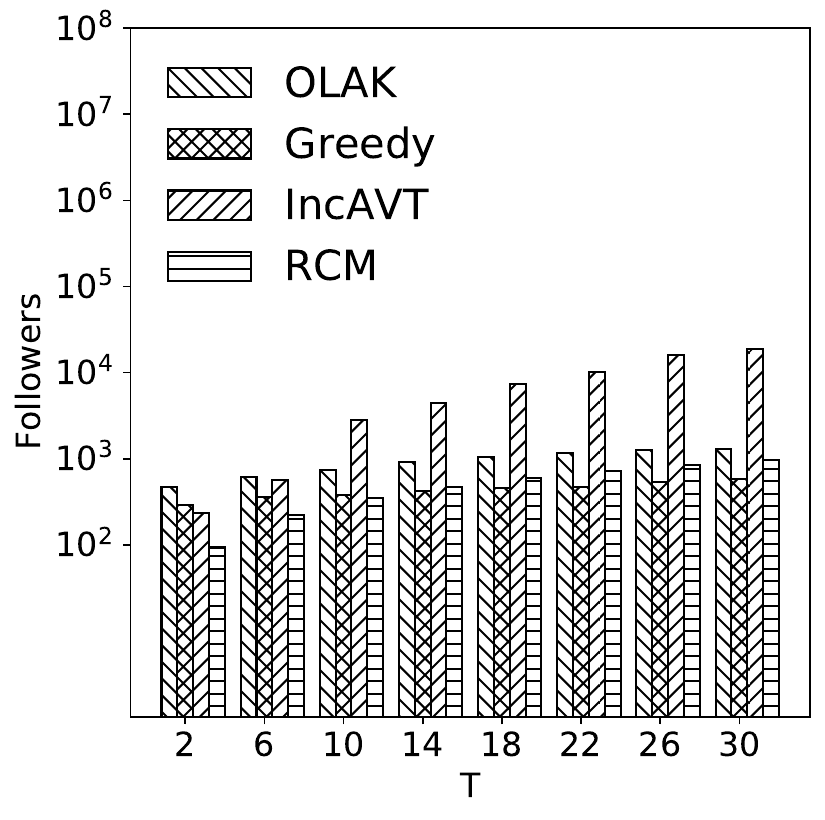}
	}
	\subfigure[Gnutella]{	\label{R3:exp_follows:vary_t2}		
		\includegraphics[width=0.3\linewidth]{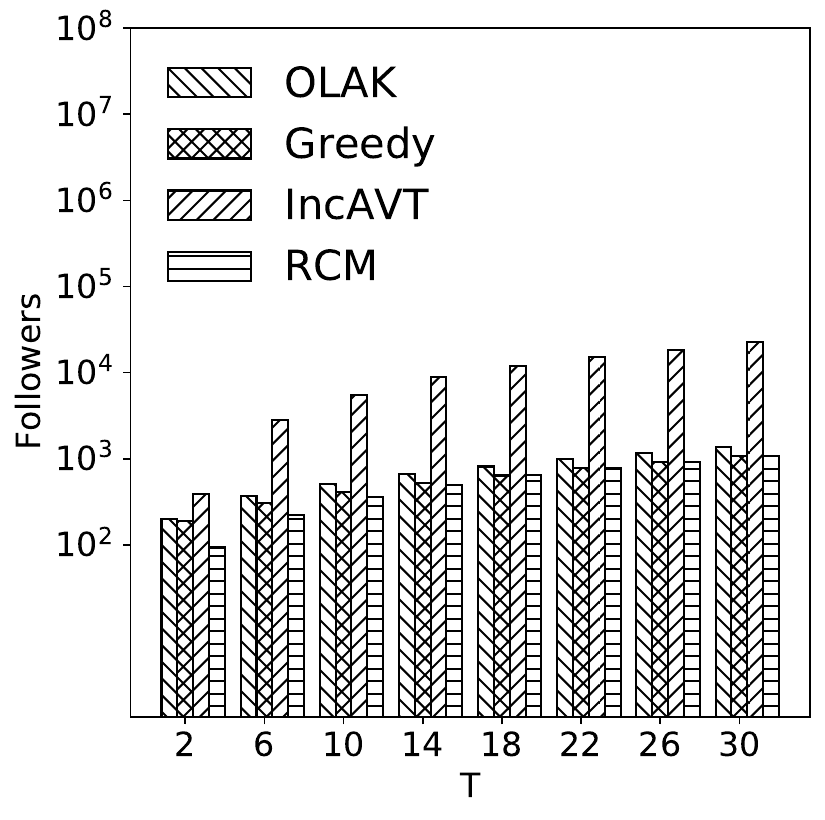}
	}
	\subfigure[Deezer]{	\label{R3:exp_follows:vary_t3}
		\includegraphics[width=0.3\linewidth]{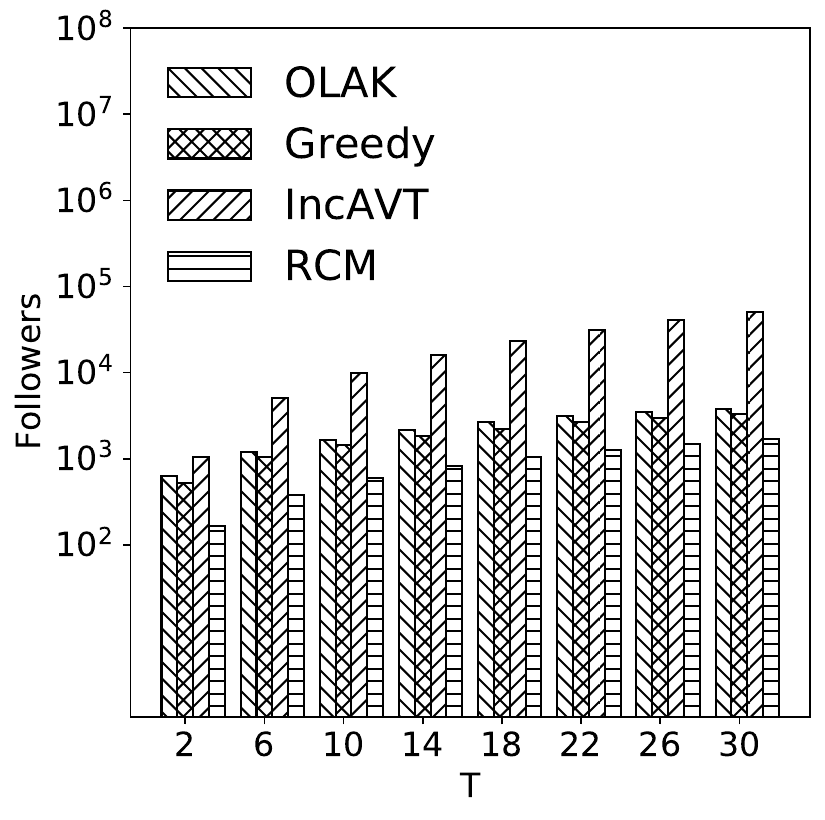}
	}
	\subfigure[eu-core]{	\label{R3:exp_follows:vary_t4}		
		\includegraphics[width=0.3\linewidth]{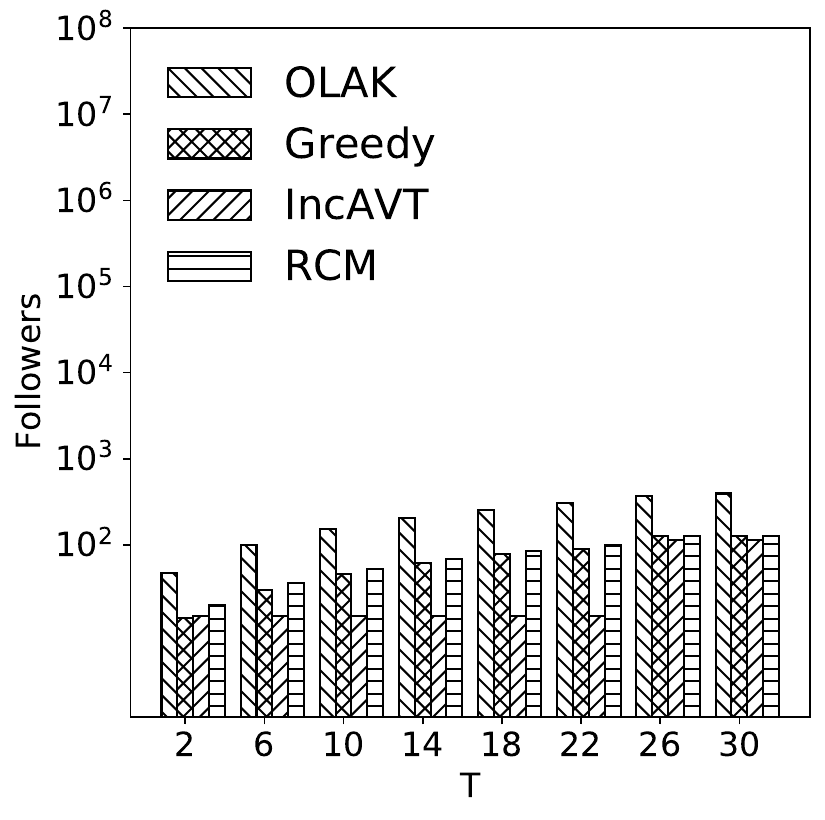}
	}
	\subfigure[mathoverflow]{	\label{R3:exp_follows:vary_t5}
		\includegraphics[width=0.3\linewidth]{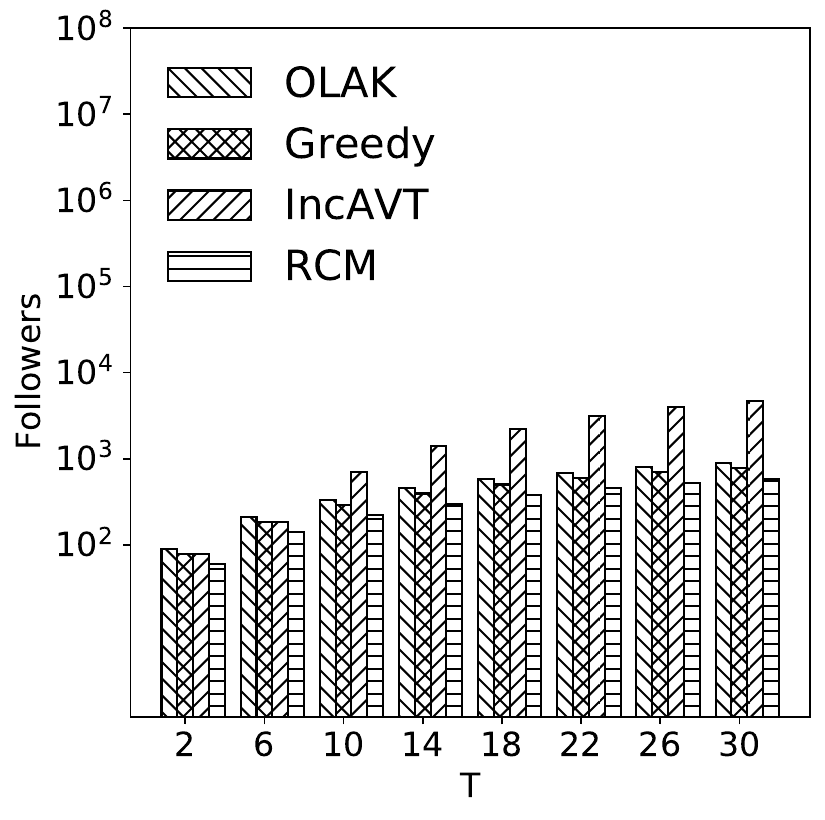}
	}
	\subfigure[\TTao{CollegeMsg}]{	\label{R3:exp_follows:vary_t6}
		\includegraphics[width=0.3\linewidth]{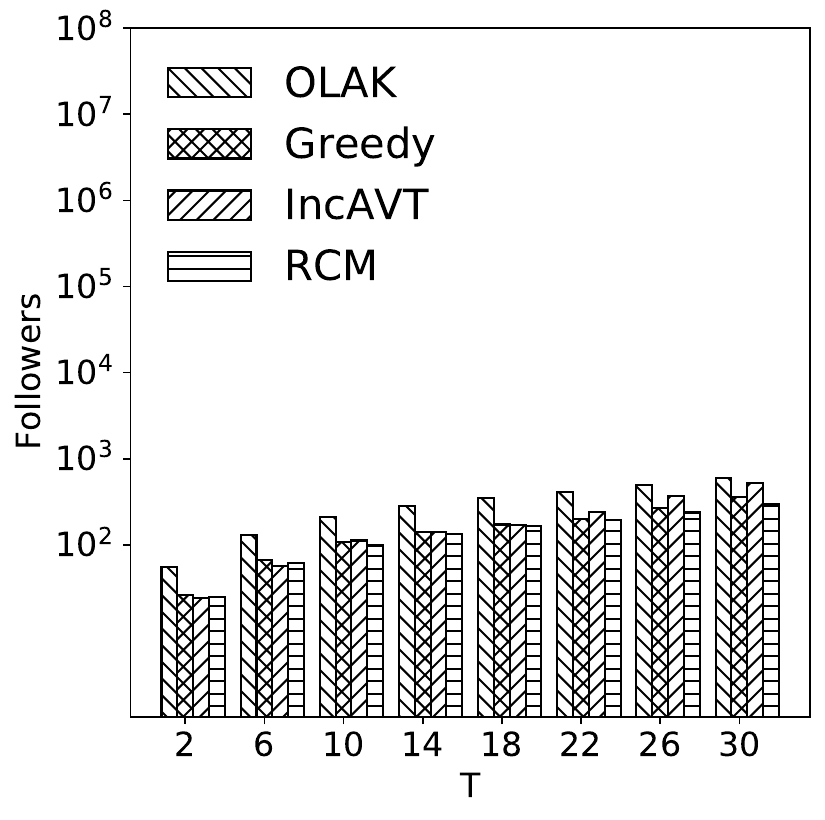}
	}
	\vspace{-2mm}
	\caption{\TT{Number of followers with varying $T$}}
	\label{fig:followers_t}
\end{figure}

\begin{figure}[ht]
	\centering
	\subfigure[email-Enron]{	\label{R3:exp_follows:vary_l1}
		\includegraphics[width=0.3\linewidth]{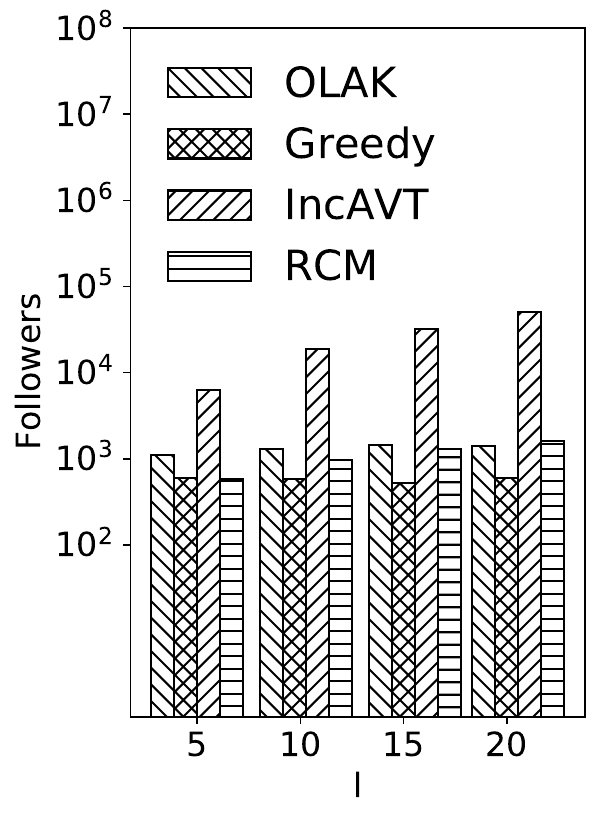}
	}
	\subfigure[Gnutella]{	\label{R3:exp_follows:vary_l2}		
		\includegraphics[width=0.3\linewidth]{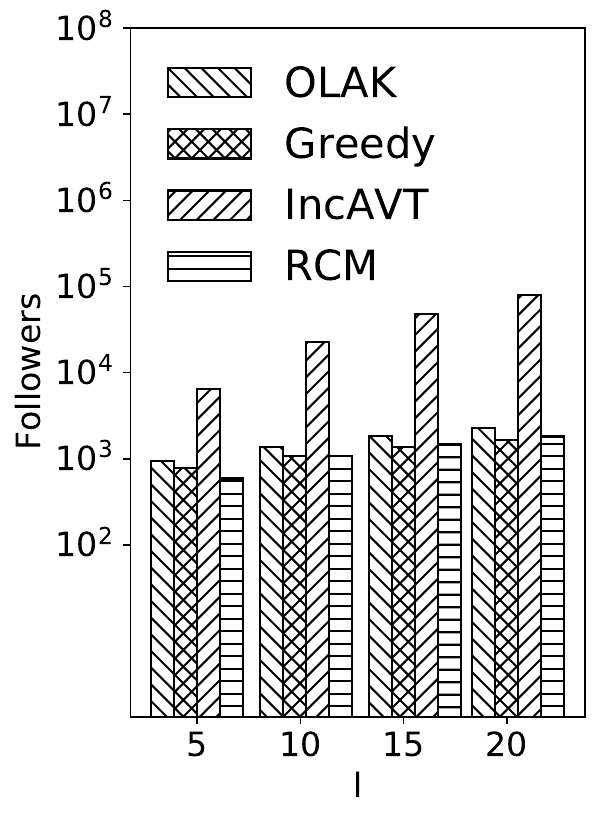}
	}
	\subfigure[Deezer]{	\label{R3:exp_follows:vary_l3}
		\includegraphics[width=0.3\linewidth]{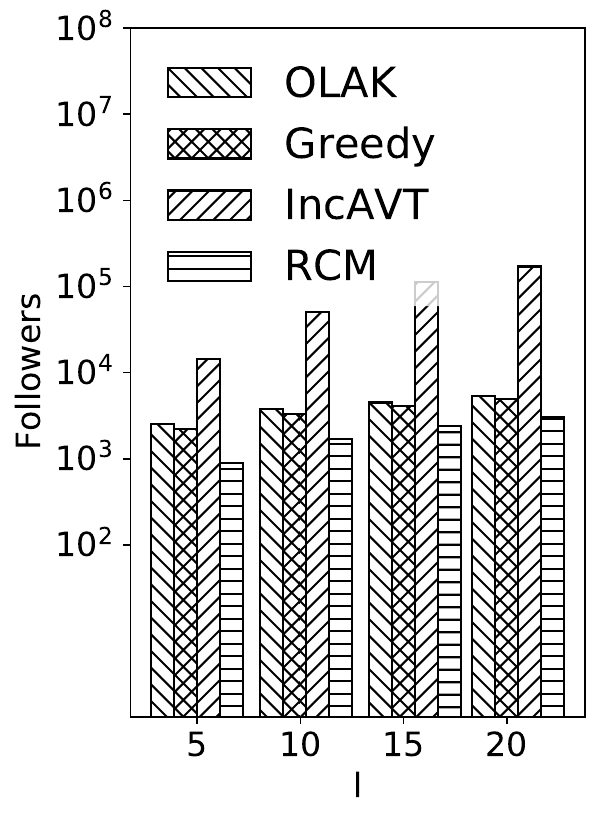}
	}
	\subfigure[eu-core]{	\label{R3:exp_follows:vary_l4}		
		\includegraphics[width=0.3\linewidth]{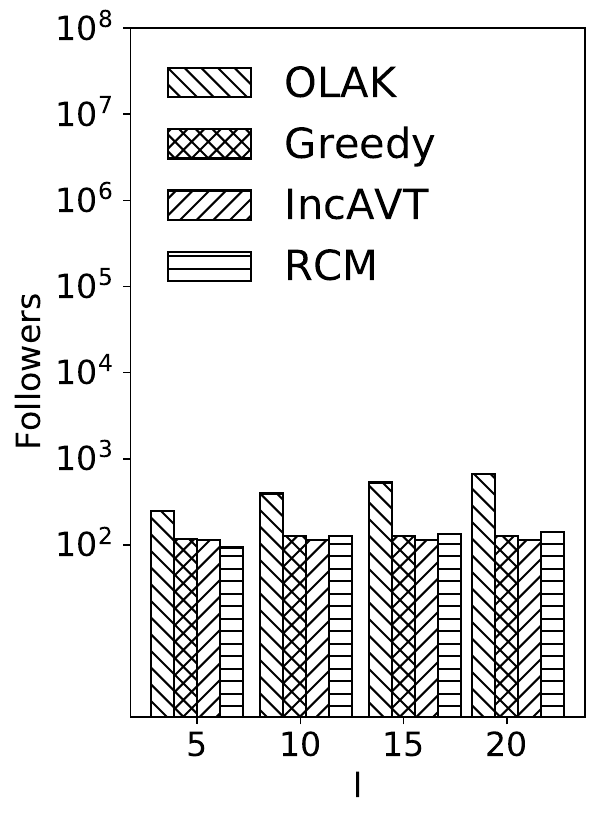}
	}
	\subfigure[mathoverflow]{	\label{R3:exp_follows:vary_l5}
		\includegraphics[width=0.3\linewidth]{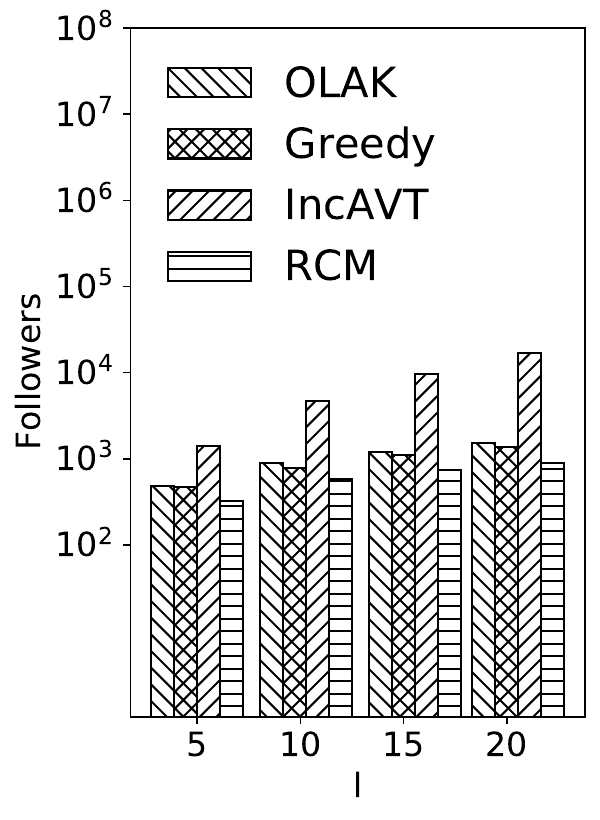}
	}
	\subfigure[\TTao{CollegeMsg}]{	\label{R3:exp_follows:vary_l6}
		\includegraphics[width=0.3\linewidth]{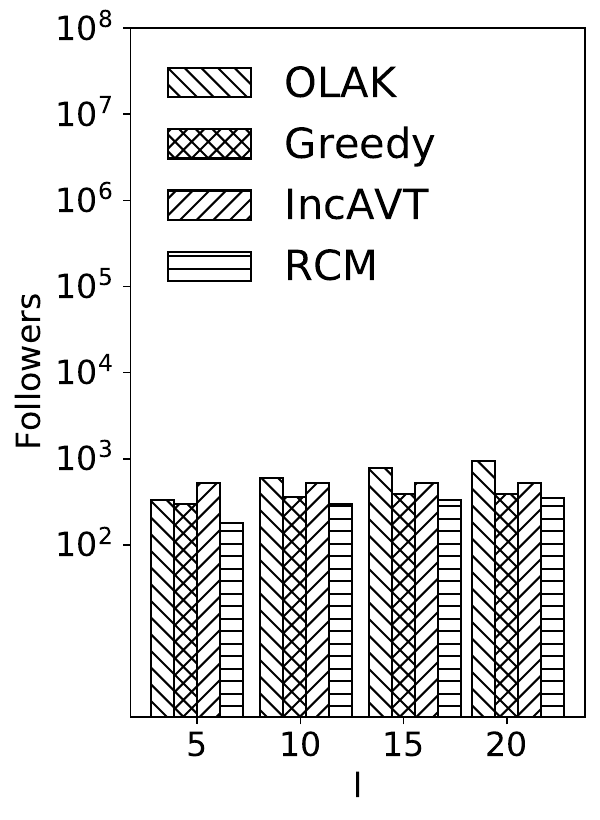}
	}

		\vspace{-2mm}
	\caption{\TT{Number of followers with varying $l$}}
	\label{fig:followers_l}
	\vspace{-3mm}
\end{figure}

\begin{figure}[ht]
	\centering
	\subfigure[email-Enron]{	\label{R3:exp_follows:vary_k1}
		\includegraphics[width=0.3\linewidth]{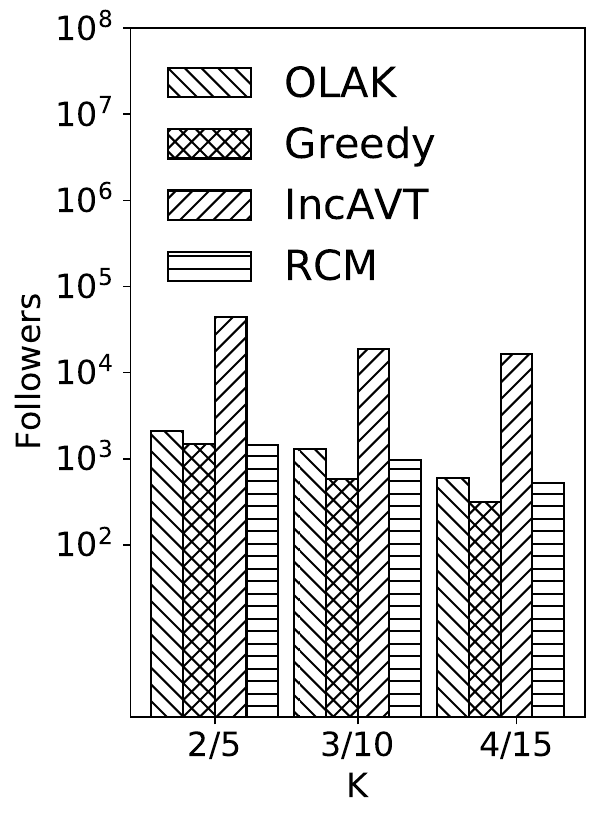}
	}
	\subfigure[Gnutella]{	\label{R3:exp_follows:vary_k2}		
		\includegraphics[width=0.3\linewidth]{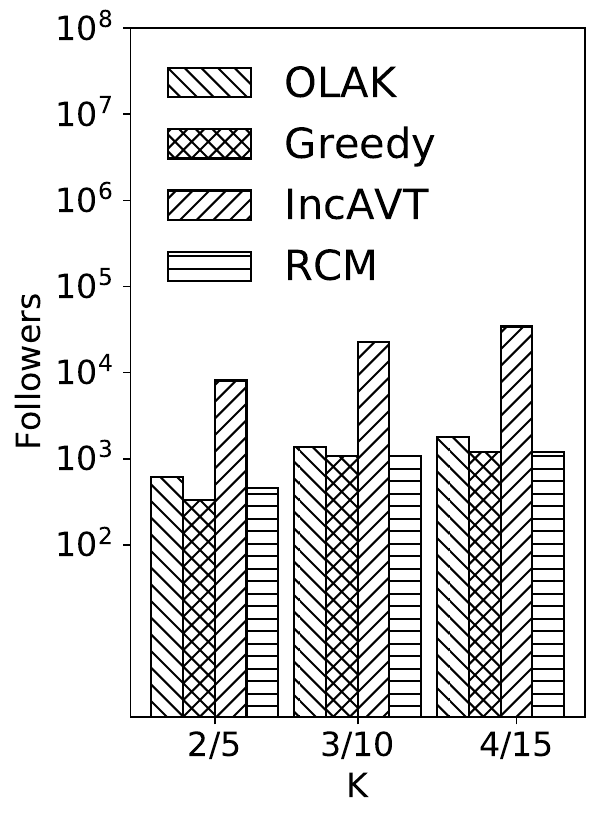}
	}
	\subfigure[Deezer]{	\label{R3:exp_follows:vary_k3}
		\includegraphics[width=0.3\linewidth]{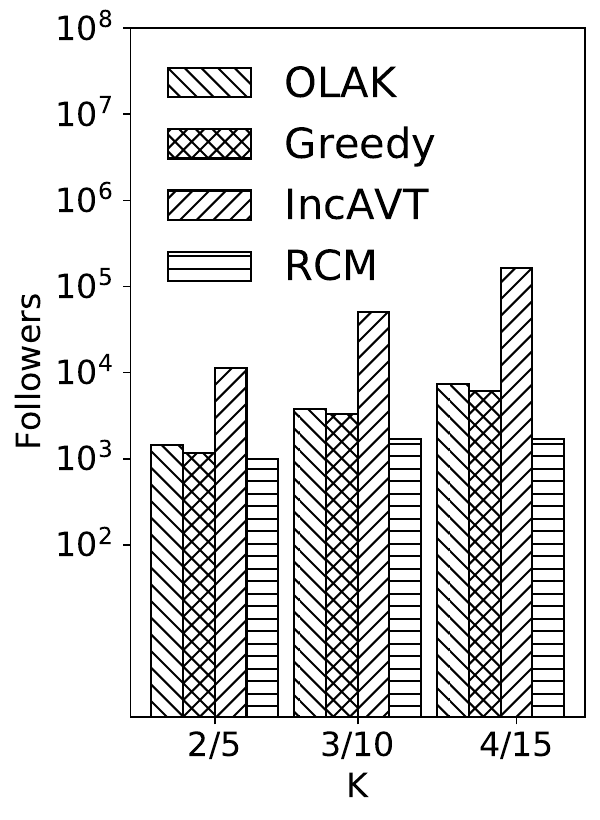}
	}
	\subfigure[eu-core]{	\label{R3:exp_follows:vary_k4}		
		\includegraphics[width=0.3\linewidth]{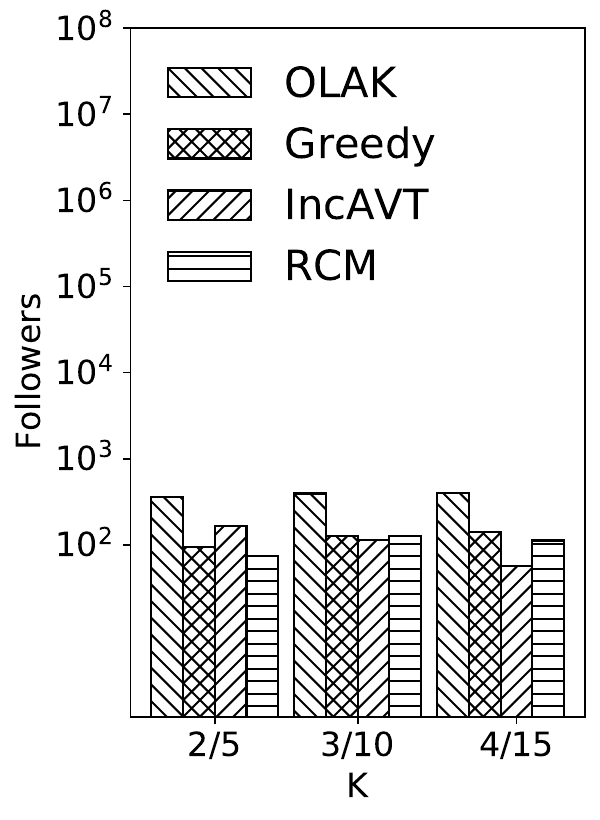}
	}
	\subfigure[mathoverflow]{	\label{R3:exp_follows:vary_k5}
		\includegraphics[width=0.3\linewidth]{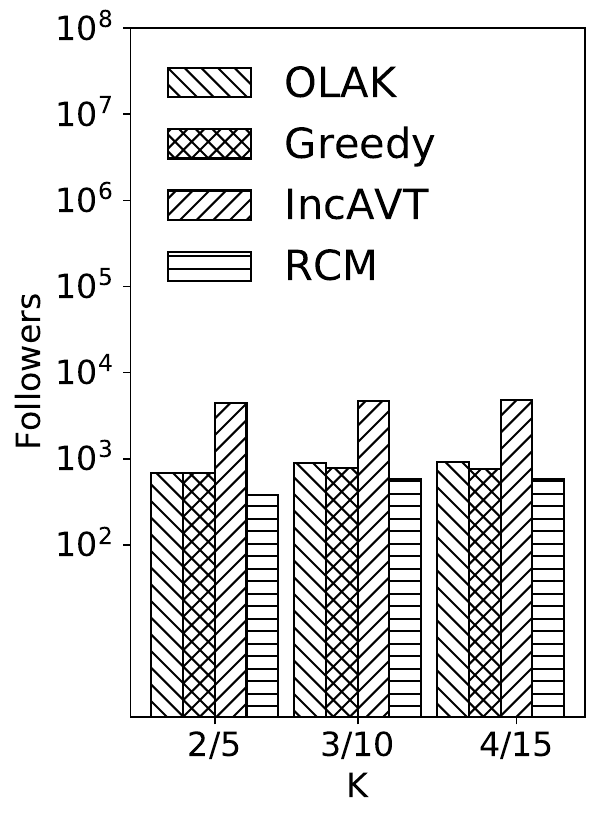}
	}
	\subfigure[\TTao{CollegeMsg}]{	\label{R3:exp_follows:vary_k6}
		\includegraphics[width=0.3\linewidth]{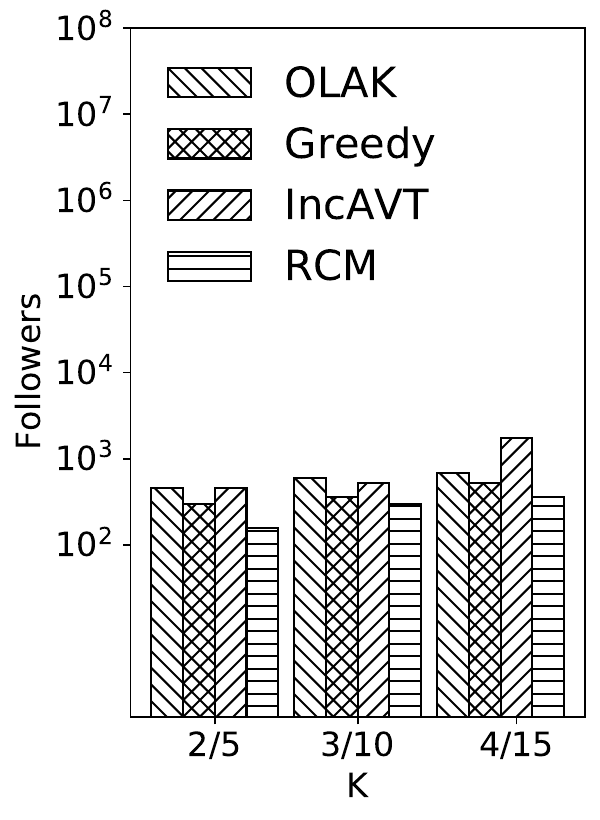}
	}
		\vspace{-2mm}
	\caption{\TT{Number of followers with varying $K$}}
	\label{fig:followers_k}
\end{figure}

In this experiment, we evaluate the total number of followers produced by the AVT problem with different datasets and approaches \TT{in Figure~\ref{fig:followers_t} - Figure~\ref{fig:followers_k}} by varying one parameter and setting the other two as defaults. 
As we can see, the number of followers in each snapshot discovered by all four approaches increases rapidly in all datasets with the evolving of the network. 
For example, \TT{in Figure~\ref{R3:exp_follows:vary_t3}}, the follower size in the \textit{Deezer} dataset is about one thousand when $T=2$ and goes up to 50,000 when $T=30$. 
\TT{Similar pattern can also be found in Figure \ref{fig:followers_l} as more followers can be found when we increase $l$ with the other two parameters fixed.} As expected, we do not observe a noticeable followers trend \TT{from Figure~\ref{fig:followers_k}} for all four approaches when varying $k$. This is because the anchored $k$-core size is highly related to the network structure. 
From the above experimental results, we can conclude that tracking the anchored vertices in an evolving network is necessary to maximize the benefits of expanding the communities.

\subsection{\TT{A Case Study on Anchored Vertex Tracking}}

\begin{figure}
    \centering
    \includegraphics[width=0.45\linewidth]{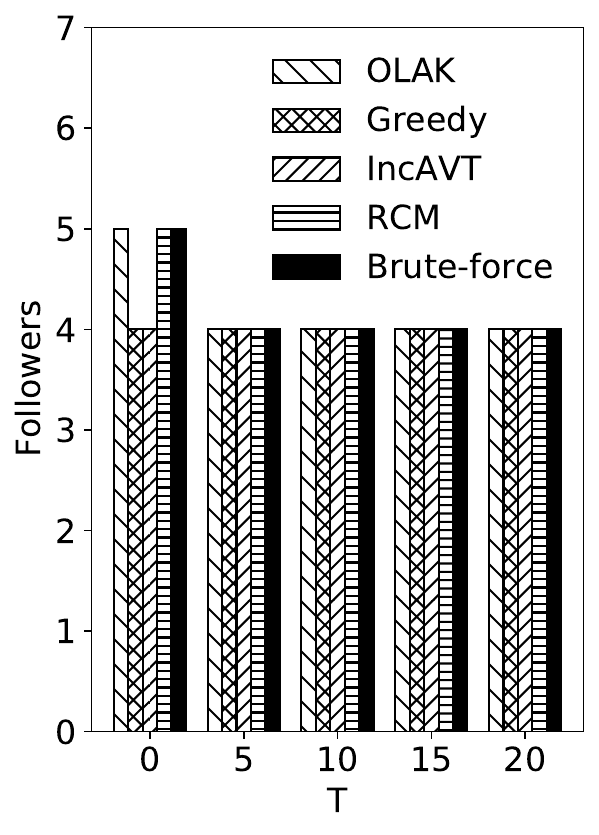}
    \caption{Follower number comparison.}
    \label{fig:my_label}
\end{figure}

\TT{We conduct a case study in this subsection to provide more insights into comparing our proposed methods with the \textit{brute-force} method for the problem studied in this paper. Specifically, the \textit{brute-force} method requires exhaustively enumerating all possible anchored sets with size $l$. The time complexity is $\mathcal{O}(C^l_{|V|}\cdot|E|)$, which is cost-prohibitive and growing exponentially while $l$ increases (\textit{e.g., the running time of \textit{brute-force} in \textit{mathoverflow} and \textit{eu-core} by setting $l=2$ and $k=3$ are over $24$ hours and 38,140 ms, respectively}).
In Figure~\ref{fig:my_label}, we report the followers results of given anchored vertices at different snapshots in \textit{eu-core} using \textit{IncAVT}, \textit{Greedy}, and \textit{brute-force} method by varying $T$ and setting $l=2$ and $k=3$. 
We observed that, the approximate results (\textit{i.e., number of followers}) reported by the four approximate algorithms (\textit{i.e., OLAK, Greedy, IncAVT, and RCM}) are  very close to the exact result queried by \textit{brute-force} algorithm. 
}

\begin{table}[t]
	\caption{Selected Anchored Vertices and Followers. \label{R3_tab:parameter}}
		\vspace{-2mm}
	\begin{center}
		\scalebox{0.9}{
			\begin{tabular}{|c|p{3.0cm}<{\centering}|c|} \hline
				{\bf Algorithms}  &  Selected Anchored Vertices                         & Followers \\\hline 
				\hline
				Brute-force              & 469, 630	                       & 163, 72, 630, 468, 469      \\\hline
				OLAK             & 630, 541     & 630, 163, 72, 541, 531   \\\hline 
				Greedy     & 541, 351                        & 541, 531, 351, 184      \\\hline
					IncAVT             &541, 351                        & 541, 531, 351, 184      \\\hline
				RCM              & 552, 630                    & 552, 630, 72, 163, 320  \\\hline
			\end{tabular}
		}
	\end{center} \label{exp:case_study2}
\end{table}

\TT{Finally, we further show the selected anchored vertices and the related followers in detail at the first snapshot period in Table~\ref{exp:case_study2}.}

\section{Related work}
\label{sec:related}
\subsection{$k$-core Decomposition} 
The model of $k$-core was first introduced by Seidman et al.~\cite{SEIDMAN1983269}, and has been widely used as a metric for measuring the structure cohesiveness of a specific community in the topic of social contagion~\cite{DBLP:journals/pnas/UganderBMK12},  user engagement~\cite{DBLP:conf/icalp/BhawalkarKLRS12, DBLP:conf/cikm/MalliarosV13}, Internet topology~\cite{DBLP:journals/nhm/Alvarez-HamelinDBV08,Carmi11150}, influence studies~\cite{Influence_Nature,DBLP:journals/amc/LiWSX18}, and graph clustering~\cite{DBLP:conf/aaai/GiatsidisMTV14,DBLP:journals/tkde/LiYM14}. 
The $k$-core can be computed by using core decomposition algorithm, while the core decomposition is to efficiently compute for each vertex its core number~\cite{DBLP:journals/corr/cs-DS-0310049}. Besides, with the dynamic change of the graph, incrementally computing the new core number of each affected vertices is known as core maintenance, which has been studied in~\cite{DBLP:conf/icde/ZhangYZQ17,DBLP:conf/sigmod/0001Y19,DBLP:journals/tkde/AksuCCKU14,DBLP:journals/tkde/LiYM14,DBLP:journals/pvldb/SariyuceGJWC13}. 


\subsection{\TTao{User Engagement}} 
User engagement in social networks has attracted much attention while quantifying user engagement dynamics in social networks is usually measured by using $k$-core~\cite{DBLP:conf/icalp/BhawalkarKLRS12, DBLP:journals/iandc/ChitnisFG16, DBLP:conf/aaai/ZhangZQZL17, DBLP:journals/tkde/ZhangLZQZ20,DBLP:conf/ijcai/ZhouZLZ019,DBLP:conf/sigmod/LinghuZ00Z20,DBLP:conf/icde/CaiLHMYS20, DBLP:conf/sdm/LaishramSEPS20}. Bhawalker et al.~\cite{DBLP:conf/icalp/BhawalkarKLRS12} first introduced the problem of anchored $k$-core, which was inspired by the observation that the user of a social network remains active only if her neighborhood meets some minimal engagement level: in $k$-core terms. Specifically, the anchored $k$-core problem aims to find a set of anchored vertices that can further induce maximal anchored $k$-core. Then, Chitnis et al. ~\cite{DBLP:journals/iandc/ChitnisFG16} proved that the anchored $k$-core problem on general graphs is solvable in polynomial time for $k \leq 2$, but is NP-hard for $k > 2$. Later, Zhang et al. in 2017 ~\cite{DBLP:conf/icde/ZhangYZQ17} proposed an efficient greedy algorithm by using the vertex deletion order in $k$-core decomposition, named OLAK. In the same year, another research~\cite{DBLP:journals/pvldb/ZhangZZQL17} studied the anchored $k$-core problem,  which aims to identify critical users that may lead a maximum $k$-core. 
Zhou et al.~\cite{DBLP:conf/ijcai/ZhouZLZ019} introduced a notion of resilience in terms of the stability of k-cores while the vertex or edges are randomly deleting, which is close to the anchored $k$-core problem. Cai et al.~\cite{DBLP:conf/icde/CaiLHMYS20} focused on a new research problem of anchored vertex exploration that considers the users' specific interests, structural cohesiveness, and structure cohesiveness, making it significantly complementary to the anchored $k$-core problem in which only the structure cohesiveness of users is considered. 
Very recently, Ricky et al. in 2020~\cite{DBLP:conf/sdm/LaishramSEPS20} proposed a novel algorithm by selecting anchors based on the measure of anchor score and residual degree, called Residual Core Maximization (RCM). The RCM algorithm is the state-of-the-art algorithm to solve the anchored $k$-core problem.
However, all of the works mentioned above on anchored $k$-core only consider the static social networks. 
\TT{Considering that the topology of networks often evolves in real-world, we proposed and studied the anchored vertex tracking problem (AVT) in this paper, which is extended from the traditional anchored $k$-core problem~\cite{DBLP:conf/icalp/BhawalkarKLRS12}, aiming to find out the optimal anchored vertices in each timestamp so as to fully maximize the community size at each period of evolving networks.}
To the best of our knowledge, our work is the first to study the anchored vertex tracking problem to find the anchored vertices at each timestamp of evolving networks.

\TTao{In addition, some other community models such as $k$-truss~\cite{huang2014querying,zhang2018finding} and $k$-plex~\cite{balasundaram2011clique} can be applied to measure the quality of user engagement dynamics in social networks. Compared with $k$-core, the $k$-truss model not only captures users with high engagement but also ensures strong tie strength among the users. However, the $k$-truss is defined based on the triangle, a local concept, and may not fully represent the user's cluster in a global view. Besides, the cohesiveness of the $k$-plex is higher than that in both $k$-core and $k$-truss. In other words, the users in $k$-plex have a tighter relationship than that in both $k$-core and $k$-truss. Nevertheless, finding a $k$-plex from a given graph for an integer $k$ is NP-hard, 
leads to the unsuitability of the k-plex model in this work.}

\section{Conclusions}
\label{sec:Con}

In this paper, we focus on a novel problem, namely 
the 
anchored vertex tracking (AVT) problem, which is the extension of the \textit{anchored $k$-core} problem towards dynamic networks. The AVT problem aims at tracking the anchored vertex set dynamically such that the selected anchored vertex set can induce the maximum anchored $k$-core at any moment. We develop a Greedy algorithm to solve this problem. We further accelerate the above algorithm from two aspects, including (1) \TT{reducing} the potential anchored vertices that need probing; and (2) \TT{proposing} an algorithm to improve the followers' computation efficiency with a given anchored vertex. Moreover, an incremental computation method is designed by utilizing the smoothness of the evolution of the network structure and the well-designed Bounded $K$-order maintenance methods in an evolving graph. 
Finally, the extensive performance evaluations also reveal the practical efficiency and effectiveness of our proposed methods in this paper.         
 
\appendices

\ifCLASSOPTIONcompsoc
  \section*{Acknowledgments}
This work was mainly supported by ARC Discovery Project under Grant No. DP200102298 and the ARC Linkage Project under Grant No. LP180100750. This work also partially supported by NNSF of China No.61972275.
\else
   regular IEEE prefers the singular form
  \section*{Acknowledgment}
\fi

\ifCLASSOPTIONcaptionsoff
  \newpage
\fi

\bibliographystyle{abbrv} 

\bibliography{Taotaocai03,Thesis}

\end{document}